\documentclass[review]{elsarticle}

\usepackage[ruled]{algorithm2e}
\usepackage{color,soul}

\SetAlFnt{\small}
\SetAlCapFnt{\small}
\SetAlCapNameFnt{\small}
\SetAlCapHSkip{0pt}
\IncMargin{-\parindent}

\usepackage{multirow}

\usepackage{algorithmic}
\usepackage{amsfonts}
\usepackage{ctable}
\usepackage{amsmath}

\usepackage{amsthm}

\newtheorem{lemma}{Lemma}

\usepackage{nth}

\usepackage{graphicx}

\usepackage{ragged2e}
\usepackage[caption = false]{subfig}

\begin{document}
 
\begin{frontmatter} 

\title{LB-OPAR: Load Balanced  Optimized Predictive and Adaptive Routing for Cooperative UAV Networks}

\author[mymainaddress]{Mohammed Gharib}
\ead{alghari@clemson.edu}

\author[mymainaddress]{Fatemeh Afghah}
\ead{fafghah@clemson.edu}

\author[mysecondaryaddress]{Elizabeth Serena Bentley}
\ead{elizabeth.bentley.3@us.af.mil}

\address[mymainaddress]{Department of Electrical and Computer Engineering, Clemson University, Clemson, South Carolina, USA. }
\address[mysecondaryaddress]{Air Force Research Laboratory Rome, NY, USA. }



\begin{abstract}


Cooperative ad-hoc UAV networks have been turning into the primary solution set for situations where establishing a communication infrastructure is not feasible. Search-and-rescue after a disaster and intelligence, surveillance, and reconnaissance (ISR) are two examples where the UAV nodes need to send their collected data cooperatively into a central decision maker unit. Recently proposed SDN-based solutions show incredible performance in managing different aspects of such networks. Alas, the routing problem for the highly dynamic UAV networks has not been addressed adequately. An optimal, reliable, and adaptive routing algorithm compatible with the SDN design and highly dynamic nature of such networks is required to improve the network performance. This paper proposes a load-balanced optimized predictive and adaptive routing (LB-OPAR), an SDN-based routing solution for cooperative UAV networks. LB-OPAR is the extension of our recently published routing algorithm (OPAR) that balances the network load
and optimizes the network performance in terms of throughput, success rate, and flow completion time (FCT). We analytically model the routing problem in highly dynamic UAV network and propose a lightweight algorithmic solution to find the optimal solution with $O(|E|^2)$ time complexity where $|E|$ is the total number of network links. We exhaustively evaluate the proposed algorithm's performance using ns-3 network simulator. Results show that LB-OPAR outperforms the benchmark algorithms by $20\%$ in FCT, by $30\%$ in flow success rate on average, and up to $400\%$ in throughput.\footnote{DISTRIBUTION STATEMENT A: Approved for Public Release; distribution unlimited AFRL-2021-1606 on 25 May 2021. Other requests shall be referred to AFRL/RIT 525 Brooks Rd Rome, NY 13441} 

\end{abstract}

\begin{keyword}
FANET, UAV, routing, SDN, load balancing, performance evaluation.
\end{keyword}

\end{frontmatter}


\section{Introduction}
\label{sec:introduction}
Cooperative unmanned aerial vehicle (UAV) networks, also known as flying ad-hoc networks (FANET), recently become extremely popular for their ease of set-up and use, low price, and high maneuverability in 3D space \cite{WalidSurvey}. In such networks, nodes work cooperatively to make the communication of two far away nodes possible by acting as intermediate repeaters. Using such networks where implementing the infrastructure is not feasible, or the infrastructure is damaged, is crucially important. Wildfire monitoring \cite{wildFire1}, search and rescue after a disaster \cite{disaster}, ISR \cite{battlefield}, ad hoc UAV cloud service provider \cite{gtss}, and border surveillance \cite{border} are some examples to name. A promising architecture for such networks uses software-defined networking (SDN) as its cornerstone to add flexibility to the network control and  management \cite{178}. SDN separates the control plane from the data plane to pass the decision-making process to a software-based controller. There are excellent efforts in designing SDN-based topology management \cite{178}, resource management \cite{180}, path planning \cite{pathPlanning}, node positioning \cite{248}, and monitoring \cite{214}. The SDN-based routing problem, of how to find the optimal path from the source node to the destination, is also tackled for multi-path routing \cite{multipath}, and low mobility UAV networks \cite{lowMob}. However, the SDN-based routing problem for highly dynamic UAV networks is still a challenging unresolved problem to the best of our knowledge. 

In this paper, we consider a cooperative UAV network with a central controller. We assume a low-bandwidth reliable communication channel exists between the controller and the UAV nodes. However, the UAVs need to transfer a high volume of data, such as high-resolution pictures, videos, or thermal pictures. Hence, they need a high bandwidth communication channel for data transfer. In this way, the control plane and the data plane are separated, just like any other SDN-based UAV network. It is worth mentioning that the controller can be a ground or aerial controller with a line of sight or non-line of sight communication with the nodes. Such communication could be borrowed from the existing cellular network, be provided by the line-of-sight communication with a controller placed at a higher elevation, or use a dedicated bandwidth for command and control, as it is the common assumption in many UAV networks \cite{WalidSurvey}. The controller is responsible for decision-making in all network management processes. We assume the routing, as one of the control management tasks, is also one of the controller responsibilities. Fig. (\ref{fig::systemModel}) represents a schematic view of our system model. Since the controller is responsible for all network management tasks, any algorithm required to be performed by the controller is desired to be with a low computation and communication complexity to prevent any potential delay in network management. This fact makes the routing problem more challenging in such a network architecture. 

\begin{figure}
    \centering
    \includegraphics[width=\columnwidth]{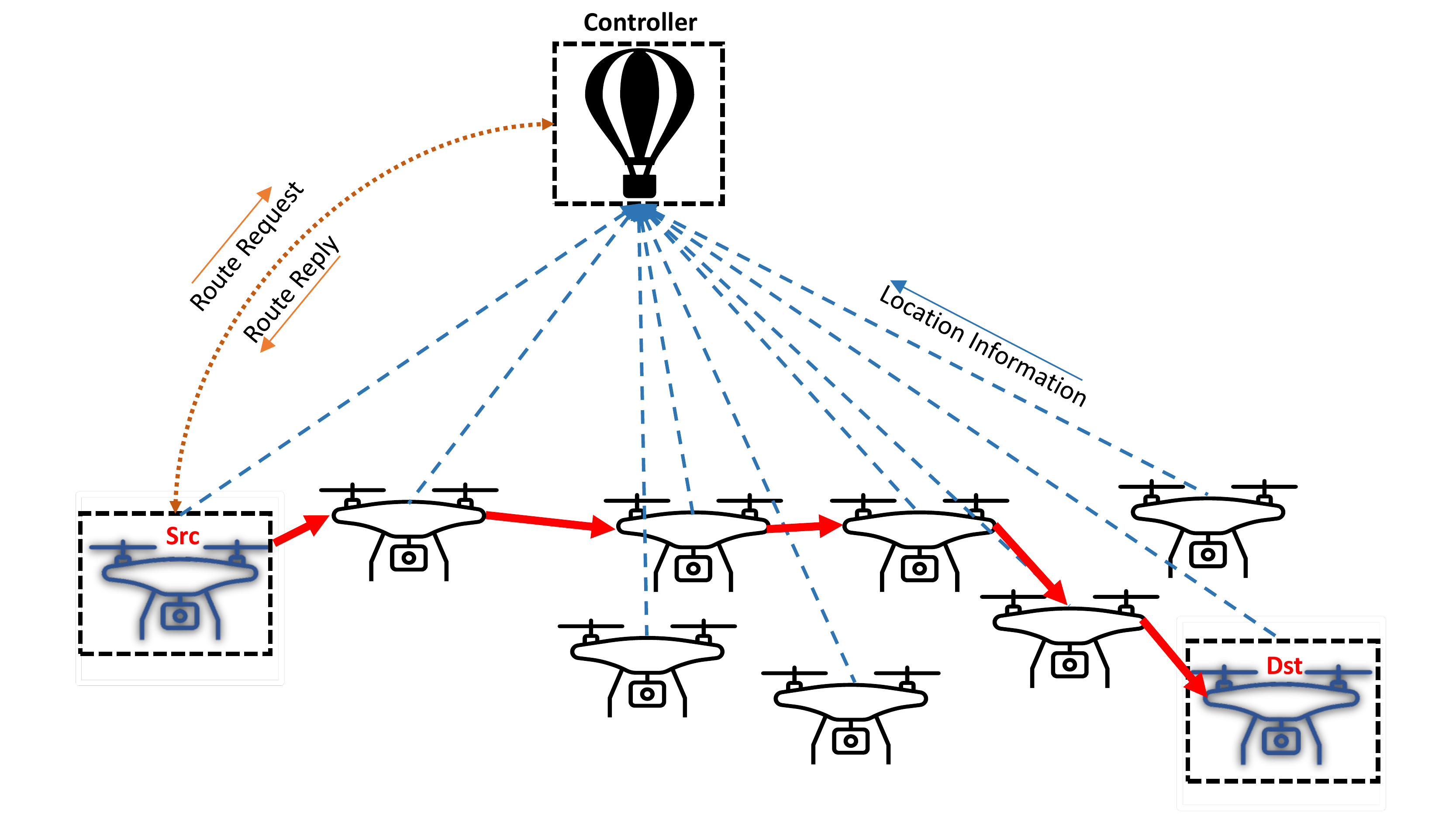}
    \caption{System model (solid lines represent a direct link where the dashed lines represent a communication which may or may not be direct)}
    \label{fig::systemModel}
\end{figure}

Cooperative UAV networks typically use the conventional routing algorithms designed for the conventional ad-hoc networks such as Mobile Ad-hoc NETwork (MANET) and Vehicular Ad-hoc NETwork (VANET). Ad-hoc On-demand Distance Vector (AODV) \cite{aodv}, Destination Sequenced Distance Vector (DSDV) \cite{dsdv}, and Optimized Link State Routing (OLSR) \cite{olsr} are some instances. Such algorithms aim at finding the shortest path between the source and the destination and do not consider the other network characteristics. Furthermore, the highly dynamic nature of UAV networks leads to fast disconnection and consequently a need for multiple reroute processes \cite{tcpreorder}. Multiple reroute processes dramatically increase the routing communication complexity as well as overall end-to-end communication delay. In this case, the data packets belonging to the same data file/stream may traverse different paths, and hence they arrive out-of-order. The same sideeffect is tackling the multi-path routing algorithms. TCP deals with some of out-of-order packets as the lost packets and orders retransmission. 

We recently published an optimized predictive and adaptive routing (OPAR) \cite{opar} for cooperative UAV networks to consider the path lifetime in addition to the path length in the path selection process. The aim was to reduce the reroute process's overheads. We find that considering the path lifetime can significantly improve UAV networks' performance in terms of throughput, flow completion time (FCT), and flow success rate. We further find that considering only the path lifetime and path length in selecting the paths may cause an unbalanced network load. Since the network uses a shared medium, unbalanced load in the network may cause performance degradation in throughput and consequently FCT, especially in the highly loaded networks. Accordingly, in this paper, we extend the OPAR algorithm to propose a load-balanced OPAR (LB-OPAR) routing algorithm by considering the network load in addition to the path lifetime and path length in selecting the paths. We show that this consideration does not impose any extra routing traffic to the network in comparison with OPAR, in the SDN-based setting. Accordingly, we model the entire problem of the SDN-based routing algorithm in a cooperative UAV network with an analytical optimization model. The proposed model is binary linear programming (BLP) problem. While BLP problems are very well-known to be  NP-complete problems, we show that in the specific case of the modeled BLP, we can find the optimized path using a graph-based algorithmic solution. The algorithmic solution is lightweight with the computational complexity of $O(|E^2|)$ where $|E|$ is the number of network links. 

LB-OPAR needs only small-size messages to be sent to the controller periodically. Each message contains three consequent positions of the UAV node. The controller predicts the link lifetime of the UAV nodes and uses the predicted values in solving the optimization problem. The prediction algorithm is lightweight, with a computation complexity of $O(|V|^2)$ where $|V|$ is the number of UAV nodes. The controller uses its own data to calculate the communication load on each node. Hence, no further communication is required. The low communication complexity and the low computational complexity of the LB-OPAR algorithm make it a perfect fit for the SDN-based model.   

To evaluate the performance of LB-OPAR, we use the ns-3 network simulator \cite{ns3}. We simulate the network to compare LB-OPAR with the well-known baseline routing algorithms AODV, DSDV, and OLSR as benchmarks. We compare the results with those of OPAR as well. We show in Section (\ref{sec:RelatedWork}) that the state-of-the-art uses these benchmarks to find the routes in the cooperative UAV networks or build new routing algorithms on the basis of these benchmarks. We selected the algorithms for comparison that cover the wide range of reactive and proactive algorithms as well as distance vector and link-state ones. We measure the network throughput, FCT, and flow success rate as the performance metrics. We show that LB-OPAR improves the network throughput by up to four times the benchmarks. It decreases the FCT for average $20\%$ and increases the flow success rate by about $30\%$.

This work's main contribution is to propose a load-balanced routing algorithm for SDN-based cooperative ad-hoc networks. More specifically, $i)$ it analytically models the routing problem in the mentioned systems with linear optimization technique and considers the path length, lifetime, and load altogether; $ii)$ it proposes an algorithmic solution for the proposed problem; $iii)$ it analytically proves the solution's complexity; $iv)$ it proposes a lightweight highly accurate link lifetime prediction method; $v)$ it exhaustively evaluates the performance of the proposed algorithm and compares it with the benchmark solutions.   

The rest of this paper is organized as follows. The state-of-the-art is reviewed in Section (\ref{sec:RelatedWork}). We propose the analytical model with its theoretical solution of the LB-OPAR in Section (\ref{sec::proposedAlg}). The performance of LB-OPAR is evaluated and compared to the baseline solutions in Section (\ref{sec:NumericalResults}). Finally, we conclude the paper in Section (\ref{sec::conclusion}).

\section{Related Work}
\label{sec:RelatedWork}
Cooperative UAV networks typically use conventional routing algorithms, which are proposed originally to be used in mobile ad-hoc networks (MANETs), vehicular ad-hoc networks (VANETs), or wireless sensor networks (WSNs). However, there are several routing algorithms dedicatedly proposed for cooperative UAV networks. Thus, in this section, we first review the conventional routing algorithms, then we present the algorithms proposed dedicatedly for UAV networks. Finally, we review the proposals for SDN-based UAV network routing algorithms to cover the state-of-the-art.

Generally, there are two categorizations for conventional ad-hoc routing algorithms \cite{adhocRoutin}. Each of them looks at the problem from a different viewpoint. However, they all are proposed to find the path with the least number of hop counts, and they do not consider other metrics such as path quality, reliability, or lifetime. The first looks at the algorithms based on their nature of being proactive, reactive, or hybrid. Proactive routing protocols periodically keep track of the network topology changes, regardless of the demand for new routes. Optimized link-state routing protocol (OLSR) \cite{olsr} and destination sequence distance vector (DSDV) \cite{dsdv} are two well-known instances of this category. In contrast to proactive protocols, reactive protocols attempt to find a route only when a new route is needed. Ad-hoc on-demand distance vector (AODV) \cite{aodv} is a very well-known instance of this category. 
The second categorization considers the method of implementing the routing algorithms to be link state or distance vector. In the link-state approach, each node builds a general view of the entire network topology to find the optimized paths. The OLSR protocol belongs to this category. In contrast, each node needs to keep only the next hop toward the destination and the path cost in the distance-vector approach. Both AODV and DSDV use this approach in their routing.




Most of the routing algorithms proposed dedicatedly for UAV networks are built based on the conventional ad-hoc routing algorithms and extend them to consider another metric in their path selection process. As an instance, Bomio et al. \cite{ON} proposed enhanced AODV (E-AODV) by taking into account a route reliability criterion. They considered a 2D area to check whether two nodes are getting closer or not. If their distance is going to be increased, the corresponding link is discarded. Li et al. \cite{LEPR} proposed link-stability estimation-based preemptive routing (LEPR) based on AODV. LEPR finds multiple disjoint paths using the AODV algorithm, calculates the expected transmission count (ETX) metric of \cite{ETX} for each link of the route, and keeps all routes to use them when the path with the best ETX becomes not available anymore. 

Pu et al. \cite{ltaolsr} proposed link quality and traffic-load aware OLSR (LTA-OLSR) to find the path with the highest quality from the set of the shortest paths found by OLSR. Song et al. \cite{olsrpdm} proposed OLSR-based mobility and delay prediction (OLSR-PMD) to choose the more stable next-hop nodes based on a prediction algorithm. Rosati et al. \cite{polsr} proposed predictive OLSR (P-OLSR) to add weight to the ETX metric of each link based on the corresponding nodes' relative speeds. Lyu et al. \cite{qngpsr} proposed  Q-Network Enhanced Geographic Ad-Hoc Routing Protocol based on GPSR (QN-GPSR). Authors reduced the premier forwarding of GPSR \cite{gpsr} algorithm by applying Q-learning, a reinforcement learning algorithm, into the GPSR routing protocol. Their results show that QN-GPSR has a higher packet delivery ratio and a lower end-to-end delay than the basic GPSR protocol. However, applying a Q-network increases the computational complexity of the GPSR algorithm by order of magnitude. Liu et al. \cite{qmr} proposed another Q-learning based Routing protocol for Flying Ad Hoc Networks and named it Q-learning based multi-objective routing (QMR). Authors show that QMR provides a higher packet delivery ratio, lower delay, and energy consumption than existing Q-learning-based routing protocols. However, the vast imposed computational complexity by the Q-learning algorithm is the main drawback of this algorithm. Adaptive and Reliable Routing Protocol with Deep Reinforcement Learning (ARdeep) \cite{ardeep} is also proposed to be applied to UAV networks. Although the authors show that it performs better than the basic GPSR algorithm in terms of average end-to-end delay and packet delivery ratio, it increases the computational complexity by order of magnitude, just like other deep learning-based algorithms. Although most of the computational complexity source of ML algorithms belongs to the training phase, a learned model might not be able to lead to high precision results in a network with a highly dynamic nature \cite{ROVIRASUGRANES2022102790}. To address this shortcoming, the learning-based methods keep the learning phase active during the network operation time to be regularly adjusted with the new live data where such online learning needs high computational complexity. 

Several other routing algorithms are proposed dedicatedly for UAV networks, but they are special-purpose algorithms or only applicable to the networks with specific missions. For instance, Sunder et al. \cite{fueling} designed a routing protocol for a UAV-based fueling system. Coelho et al. \cite{redefine} designed a routing algorithm for UAV networks where UAVs are used as WiFi repeaters to maximize the coverage. The work of \cite{packageDelivery} designed for package delivery networks. Cheriguene et al. \cite{sempr} proposed a Swarm energy-efficient multicast routing protocol for UAVs flying in group formations, to facilitate the control and information delivery among the UAVs while minimizing inter-UAV packet loss, packet re-transmission, and end-to-end delay.
It is obvious that such mission-specific routing algorithms cannot be used as general-purpose routing for all cooperative UAV networks.   

An SDN architecture is proposed to be used in UAV networks to increase the network control flexibility and improves the network performance. Different architectures are proposed for SDN-based UAV networks \cite{191,qos,infocomwkshp2018,SDNtraffic}. All the proposed architectures have a controller or a set of distributed controllers that has/have direct, reliable communication with the UAV nodes. Using such an architecture, any of the network management tasks could be migrated into the controller. However, the migrated tasks do not impose huge computational or communication complexity into the controller. Routing, as an essential network management task, could be migrated into the controller. To the best of our knowledge, the problem of finding an optimal path in an SDN-based UAV network has not yet been substantially studied. Several recently published works mentioned this problem \cite{lowMob,multipath,SDNtraffic}.

Ramaprasath et al. \cite{lowMob} proposed an SDN-based routing algorithm that utilizes AODV in the initialization phase of the network. Alas, this proposal cannot capture the high dynamicity of the UAV networks due to its multiple long wait times during the routing procedure. Hence, this algorithm's networks of interest are the UAV networks with deficient mobility. Authors of \cite{multipath} proposed an SDN-based routing algorithm that finds all possible paths toward the destination and divides the entire data file among them. Each part of the data file is then sent through one of the paths. As the out-of-order packet arrival is one of the main challenges of the highly dynamic networks, this works increases the number of such packets by using multiple paths. Authors also assume that each UAV is equipped with different network interfaces. Hence, it can send the data file part from different paths concurrently, which poses a hardware limitation for networks aimed at using this algorithm. Qi et al. \cite{SDNtraffic} proposed traffic differentiated routing (TDR), a new SDN architecture for UAV networks that includes two types of controllers, cluster controllers and a coordination controller. In this architecture, the network is clustered, and each cluster needs a controller. To arrange different cluster controllers, a coordination controller is also required. The routing algorithm of this paper is proposed to find the routes inside the clusters. Besides their new network architecture, authors assumed that the controller is aware of the current location and the nodes' speed. Using such information, the authors proposed an optimization problem in which finding the optimal solution is not practically feasible. Hence, they use approximation methods to find a near-optimal path. As we mentioned earlier, the routing problem in SDN-based cooperative UAV networks needs substantial work. Accordingly, we propose LB-OPAR as a general-purpose routing algorithm for SDN-based UAV networks in this paper. Table (\ref{tbl::relatedWork}) represents a comparison between routing algorithms proposed dedicated for UAV networks or could be used in them, including LB-OPAR.

\begin{table}[t]

\caption[]{Routing Algorithms for Cooperative UAV Networks. (The comparison of this table is partially based on the results of  Section (\ref{sec:NumericalResults}) of this paper.)}
\resizebox{.7\columnwidth}{!}{
\begin{minipage}{\textwidth}
\begin{tabular}{ | l | l | l | l | l | l | l | l | }
\hline
Algorithm  &Mobility &Routing  &Commun.&Comput. &SDN-based/ &Load  &Success\\
&Support &  Metric &  Overhead&  Overhead&Fully-dist. & Balancing&Rate\\\hline
AODV\cite{aodv} &Good &Path len.&High &Low & Fully-dist.&$\qquad\times$&Medium \\
DSDV\cite{dsdv} & VLimited &Path len.  &VHigh &Low &Fully-dist. &$\qquad\times$&VLow\\
OLSR\cite{olsr}& VLimited &Path len.&VHigh &Low &Fully-dist. &$\qquad\times$&VLow\\
E-AODV\cite{ON}&Good &Distance&High & Low& Fully-dist.&$\qquad\times$&Low\\
LEPR\cite{LEPR}&Good &ETX&High & Low &Fully-dist. &$\qquad\times$&Medium \\ 
LTA-OLSR\cite{ltaolsr}&VLimited &RSSI+load  &VHigh &Low & Fully-dist.&$\qquad\checkmark$&VLow\\
OLSR-PMD\cite{olsrpdm}&VLimited &Delay  &VHigh &Low &Fully-dist. &$\qquad\times$&VLow\\
P-OLSR\cite{polsr}&VLimited &Weighted ETX  &VHigh &Low & Fully-dist. &$\qquad\times$&VLow\\
QN-GPSR\cite{qngpsr}&Limited &Link Quality  &VHigh & VHigh& SDN-based&$\qquad\times$&Low\\
ARdeep\cite{ardeep}&Medium & PER\footnote{Packet error rate} &-&VHigh & SDN-based&$\qquad\times$ &-\\
QMR\cite{qmr}&Limited & Delay \& Energy &VHigh & VHigh& SDN-based&$\qquad\times$&Low\\
\cite{multipath}&VLimited &Multipath&VHigh &Low & SDN-based&$\qquad\times$&Low\\
\cite{lowMob}&VLimited &Path Len.  &High &VHigh & SDN-based&$\qquad\times$&VLow\\
TDR\cite{SDNtraffic}&VLimited &Link Availability&- &VHigh & SDN-based&$\qquad\times$&-\\
OPAR\cite{opar} &VGood &Path len.  &Low & VLow&SDN-based &$\qquad\times$&VHigh\\
& & \& lifetime & & &&&\\
LB-OPAR &VGood & Path len.,  &VLow &Low &SDN-based &$\qquad\checkmark$&VHigh\\
& & load \& lifetime & & &&&\\\hline
\end{tabular}
\label{tbl::relatedWork}
\end{minipage}}
\end{table}

\section{Proposed Algorithm}
\label{sec::proposedAlg}

We model the routing problem in SDN-based cooperative UAV networks with a linear programming (LP) optimization model. The optimization model considers path length, path lifetime, and network load altogether. We define the path lifetime as the minimum lifetime of all the links forming the path. Hence, we need to know the link lifetime of each network link. Obviously, the link lifetime depends mostly on the nodes' movement.  Hence, we propose a link lifetime prediction algorithm to support LB-OPAR. We further consider the path load as the maximum node load for the nodes participating in the path. We define the node load as the number of the neighboring nodes engaged with data transmission, affecting the medium access accessible with the corresponding node. In other words, the node load is the number of nodes, in the corresponding node's adjacency, that are sending or forwarding data packets, i.e. they are a source of a data flow or engaged as forwarding nodes in a communication flow. Since the controller manages the routes and knows the nodes' location, it can calculate nodes' load. The proposed optimization problem variables are Binary, which means that the model is a Binary LP problem. The Binary LP problems are well-known to be non-polynomial problems in the NP-complete category \cite{karp}. However, we show that the proposed problem is an especial case in which its solution could be found using the graph-based algorithmic method.   

In this section, we propose the LP problem, considering that we know the links' lifetime and nodes' load. We then propose the graph-based algorithmic solution, which can find the solution of the Binary LP problem. Finally, we present the details of how to predict the link lifetime and calculate the node load by the controller. Table (\ref{tbl::notation}) represents the notations used throughout this paper.

\begin{table}[t!]
\caption[]{Table of notations }
\resizebox{1\textwidth}{!}{
\begin{minipage}{\textwidth}
\begin{tabular}{ l | l  }
  $n$ & Number of UAVs\\ 
  $u_i$ & The $i^{th}$ UAV\\   
  $G(V,E)$ & Network Graph\\
  $V$ & Set of graph vertices\\
  $E$& Set of graph edges\\
  $e(i,j)$ & The edge from node $u_i$ to node $u_j$\\
  $\mathcal{T}$& Link lifetime prediction matrix\\
  $\mathcal{L}$& The vector of nods' load.\\
  $\tau_{(i,j)}$& Link lifetime of $e(i,j)$\\
  $l_i$& The load of node $u_i$\\ 
  $x_{(i,j)}$& Binary variable corresponding to $e(i,j)$\\
  $T$& The inverse of path lifetime\\
  $w_1$ & Path length weight\\ 
  $w_2$ & Path lifetime weight\\
  $w_3$ & Path load weight\\ 
  $\varsigma$& The matrix of all lifetime-load costs\\
  $\varsigma_{(i,j)}$& the lifetime-load cost of $e_{(i,j)}$\\
  $\varsigma_{path}$&The maximum lifetime-load cost of a path\\ 
  $x_i(t_0)$&The position of $u_i$ at time $t_0$ along x-axis in 3D space\\
  $y_i(t_0)$&The position of $u_i$ at time $t_0$ along y-axis in 3D space\\
  $z_i(t_0)$&The position of $u_i$ at time $t_0$ along z-axis in 3D space\\  
  $\alpha_i$& The azimuthal degree of node $u_i$ direction\\
  $\theta_i$& The polar degree of node $u_i$ direction\\
  $v_i(t_0)$& The speed of node $i$ in time $t_0$\\
  $a_i(t_0)$& The acceleration of node $i$ in time $t_0$\\
  $d_{(t_0)}^{(i,j)}$& Euclidean distance between $u_i$ and $u_j$ at time $t_0$\\
    $d_{(t_0,t_1)}^{(i)}$& Euclidean distance traversed by node $u_i$ in time $[t_0\quad t_1]$\\
  $\Re$& Communication range
\end{tabular}
\label{tbl::notation}
\end{minipage}}
\end{table}

\subsection{LB-OPAR: SDN-Based load balanced Optimized Routing}
\label{sec::optimization}

On the one hand, it is generally a challenging problem to consider multiple metrics in defining an optimization problem. Multi-objective linear programming (MOLP) is a trivial modeling technique in such cases. MOLP solution, if feasible, leads in a set of all efficient extreme points, i.e. all maximal efficient faces. The optimal solution, then, is based on the weights of importance of different objective functions. On the other hand, finding a path requires an optimization problem with integer or Binary variables, where solving the optimization problems, in most cases, results in non-integer solutions. Rounding the non-integer solutions to the closest integer value leads to a non-optimal solution. Finding the optimized solution of most of the proposed optimization problems in routing and path planning is not feasible \cite{SDNtraffic,fueling,packageDelivery}. Accordingly, the optimization problem's design has to be such that finding the optimal solution becomes feasible, if possible. Hence, we carefully designed the following optimization problem, which is a Binary LP problem. In this problem, we assume that we have a matrix of all links' lifetimes, i.e. $\mathcal{T}$, as well as the vector of all nodes' load, i.e. $\mathcal{L}$.

\begin{eqnarray}
 \min_{} & & \sum\limits_{\substack{(i,j)\in E\\i\in\{1,\cdots,n\}\\j\in\{1,\cdots,n\}}}  w_1 x_{(i,j) } + w_2 T +  w_3  L
 \label{TimeObj} \nonumber\\
   \mbox{Subject To:} & &  \sum\limits_{(s,i)\in E} x_{(s,i)} = 1 \label{Const1} \\
 & & \sum\limits_{(i,s)\in E} x_{(i,s)} = 0 \label{Const2} \\
& &\sum\limits_{(i,d)\in E} x_{(i,d)} = 1 \label{Const3} \\
& & \sum\limits_{(d,i)\in E} x_{(d,i)} = 0 \label{Const4} \\
& & \sum\limits_{\substack{(i,j)\in E\\ i\neq s }} x_{(i,j)} = \sum\limits_{\substack{(j,k)\in E \\ j\neq d}} x_{(j,k)} \label{Const5} \\
 && T \geq \frac{x_{(i,j)}}{\tau_{(i,j)}} \label{Const6}\\
 & & x_i =\sum\limits_{\substack{j=0\\ j\ne i}}^{n} x_{(i,j)} \label{Const7}\\
 && L\geq x_i l_i \label{Const8}\\
 & & x_{(i,j)} \in \{0,1\} \label{Const9}\\
  & & 0 \leq T \leq1 \label{Const10}
\end{eqnarray}

The objective function of this problem is to minimize the path length as well as path load, and maximize the path lifetime. The significance of each of the path length, lifetime, and load is set by setting the weights $w_1$, $w_2$, and $w_3$, respectively, where $\sum_{i=1}^{3}w_i=1$. The Binary variables $x_{(i,j)}$s represent whether the corresponding network link $e_{(i,j)}$ participate in the optimal path or not. If the link $e_{(i,j)}$ is selected to be a part of an optimized path, $x_{(i,j)}$ is set to one, elsewhere it is set to zero. In this objective function, $T$ and $L$ represent the inverse path lifetime, and the path load, respectively. While $T$ has to be minimized according to the objective function, it has to be greater or equal to $\frac{x_{(i,j)}}{\tau_{(i,j)}}$ for all links $x_{(i,j)}$s, according to Constraint (\ref{Const6}). It is the idea to find the path lifetime, i.e. the minimum link lifetime for the links participated in the path, and consider it in the objective function. Similarly, the path load $L$ has to be minimized. To find the path load we minimize $L$ in the objective function where in Constraint (\ref{Const8}) we force $L$ to be greater than or equal to $x_il_i$, where according to Constraint (\ref{Const7}) $x_i =\sum_{\substack{j=0, j\ne i}}^{n} x_{(i,j)}$, and $l_i$ is the network load of node $i$.      

We need Constraints (\ref{Const1}) and (\ref{Const2}) to let the source node start the path and prevent it from being in a loop. Similarly, we need Constraints (\ref{Const3}) and (\ref{Const4}) to let the destination node be chosen as the last node in the path and terminate the path selecting operation. Constraint (\ref{Const5}) assures that the intermediate node, which receives an active link, has to have an outgoing active link. As we discussed earlier, the proposed problem variables are Binary, which is shown in Constraint (\ref{Const9}). Finally, we need Constraint (\ref{Const10}) in practice to remove the links with a very short lifetime. How short the link could be to be discarded depends on the metric chosen for link lifetime to be second, millisecond, etc.  

\subsection{Optimization Problem Solution}

Generally, Binary LP problems are in the NP-complete category \cite{karp}, and hence, they do not have a polynomial-time solutions. In this section, we show that the LB-OPAR problem is a special case of Binary LP problems where the optimal solution can be found via graph-based algorithmic solution. Needless to say, the proposed algorithmic solution might not be applicable to other Binary LP problems. While the solution is as follows, Algorithm (\ref{alg::LB-OPAR}) represents the pseudocode of the algorithm. We prove through Lemma (\ref{lem::opt}) that the output of Algorithm (\ref{alg::LB-OPAR}) is the optimal solution of the Binary LP problem proposed in section (\ref{sec::optimization}). The computation and space complexity of Algorithm (\ref{alg::LB-OPAR}) is also investigated in Lemma (\ref{lem::complexity}).

To find the solution to the proposed optimization problem, i.e. the optimal path, we consider the network as a graph $G(V,E)$ where $V$ and $E$ are the sets of graph vertices and edges, respectively. We consider any UAV node as a vertex in graph $G$. Any two neighbor UAVs which are located at the communication range of one another are shown by an edge in graph $G$. Thus, the edge $e_{(i,j)}$ in graph $G$ represents a physical link between UAV nodes $u_i$ and $u_j$. If the communication range between two nodes is not identical, we can consider $e_{(i,j)}$ as a directional edge starting from node $u_i$ toward node $u_j$, without loss of generality.   


\begin{algorithm}
\caption{Finding the solution of the optimization problem}
 \label{alg::LB-OPAR}
\begin{algorithmic}
\STATE{$objValue=\infty$}, \COMMENT{Set the initial objective value equal to $\infty$.}
\STATE{$path=\emptyset$}, \COMMENT{Set the initial path to empty path.}
\STATE{$\varsigma\leftarrow w_2\frac{1}{\mathcal{T}}+w_3\mathcal{L} $}, \COMMENT{Form $\varsigma$, a $n\times n$ matrix containing the lifetime-load cost of all graph edges. The lifetime-load cost for each edge $e_{(i,j)}$ is calculated as $\varsigma_{(i,j)}=w_2\frac{1}{\tau_{(i,j)}}+w_3l_j$.  }
\STATE{$\varsigma_{Sorted}=descentSort(\varsigma)$}, \COMMENT{Sort the lifetime-load cost matrix in descending order.}
\STATE{$G'(V,E)=G(V,E)$}, \COMMENT{Let $G'(V,E)$ equal to the main graph $G(V,E)$.}
\WHILE{$isPath(G'(V,E),src,dst)$} 
\STATE{$newPath=BFS(G'(V,E),src,dst)$}, \COMMENT{To find the shortest path with BFS algorithm.}
\IF{$newObjValue < objValue$}
\STATE{$path=newPath$},\COMMENT{Set the path equal to the new path.}
\STATE{$objValue=newObjValue$},\COMMENT{Set the objective value equal to the objective value of the new path.}
\ENDIF
\STATE{$\varsigma_{path}	\leftarrow lifetimeLoadCost(newPath)$}, \COMMENT{Calculate the maximum path lifetime-load cost for the edges forming the $newPath$.}
\STATE{$G'(V,E) \leftarrow remove(\varsigma_{path}, \varsigma_{Sorted},G'(V,E))$}, \COMMENT{Remove all links with the lifetime-load cost greater than or equal to the $\varsigma_{path}$.}
\ENDWHILE

\end{algorithmic}
\end{algorithm}


To find the optimal path, we first define a lifetime-load cost metric for each edge. We name this metric $	\varsigma_{(i,j)}$ for edge $e_{(i,j)}$ and we can calculate this metric as $	\varsigma_{(i,j)}=w_2\frac{1}{\tau_{(i,j)}}+w_3l_j$. The basic idea of Algorithm (\ref{alg::LB-OPAR}) is to use the breadth-first search (BFS) algorithm \cite{bfs} to find the shortest path in graph $G(V,E)$, regardless of its lifetime or load. BFS is the algorithm with the lowest computational complexity to find the shortest path in a graph. We then calculate the objective value for the path and the lifetime-load cost, i.e. $	\varsigma_{(i,j)}$, of all edges forming the path. Next, we remove all graph edges with less or equal lifetime-load cost than the maximum $\varsigma_{(i,j)}$ of the path. We apply the BFS algorithm again on the new graph. If the new path has a lower objective value, we consider it as the optimal path; elsewhere, we discard the new path. The procedure is repeated until no more paths stay between the source and the destination nodes. Algorithm (\ref{alg::LB-OPAR}) shows the detailed pseudocode of the proposed algorithm. We show in Lemma (\ref{lem::opt}) that the output of Algorithm (\ref{alg::LB-OPAR}) is the optimal path. We then calculate the space and computational complexity of this algorithm in Lemma (\ref{lem::complexity}).

\begin{lemma}
\label{lem::opt}
If there is more than one path from the source node toward the destination, Algorithm (\ref{alg::LB-OPAR}) returns the optimal path.
\end{lemma}
\begin{proof}
We prove this lemma using proof by induction. If there is at least one path, we know that the BFS algorithm will find it at the first iteration of the algorithm, and it is the optimal path. Assume that the algorithm runs $(n+1)$ iterations to return the final answer. Again, assume that after $n^{th}$ iteration, path $p$ is returned as the optimal path, and it is optimal until then. To prove the lemma, we need to show that the output of the last iteration, i.e. $(n+1)^{th}$ iteration, will be optimal, too. In this case, we have to show that the link(s) that we removed from the graph in $n^{th}$ iteration did not lead to discarding the optimal path. We use proof by contradiction to prove it. Hence, we assume that there exists a path $\bar{p}$ with an objective value less than that of $p$, which at least one of its links is discarded in $n^{th}$ iteration. It is worth mentioning that, since Algorithm (\ref{alg::LB-OPAR}) first finds the shortest path and then calculates the lifetime-load cost and discards some links, $\bar{p}$ cannot be shorter than $p$. If one of $\bar{p}$ links is discarded, it means that the link has a lifetime-load cost more than or equal to the maximum lifetime-load cost of $p$'s links. When $\bar{p}$ is not shorter than $p$, all of its links have to have lifetime-load costs less than that of the maximum cost of $p$'s links to become the optimal path. We meet a contradiction; hence, the lemma is proved. If there exists more than one path from the source node toward the destination, Algorithm (\ref{alg::LB-OPAR}) returns the optimal path. 
\end{proof}

\begin{lemma}
\label{lem::complexity}
The computational complexity of Algorithm (\ref{alg::LB-OPAR}) is $O(|E|^2)$ for connected networks, i.e. there is at least one path between any pair of nodes. This algorithm's space complexity is also $O(|V|+|E|)$.
\end{lemma}
\begin{proof}
We know that the computational complexity of the BFS algorithm is $O(|V|+|E|)$ where $|V|$ and $|E|$ stand for the number of graph vertices and edges, respectively \cite{bfs}. In Algorithm (\ref{alg::LB-OPAR}), in the worst case, we need to perform BFS $|E|$ times. We further need to sort the lifetimes-load cost matrix. The worst-case complexity of sorting the lifetime-load cost matrix is $O(|E|\log|E|)$. Hence, the entire algorithm complexity is $O(|E|(|V|+|E|)+|E|\log|E|)=O(|E|^2+|E||V|+|E|\log|E|)$. We know that $O(|E|)\ge O(|V|)$ for the connected networks. Hence, the computational complexity of Algorithm (\ref{alg::LB-OPAR}) is $O(|E|^2)$.
We further know that the space complexity of the BFS algorithm is $O(|V|)$. We need further $O(|E|)$ space to store the lifetime-load cost matrix. The space complexity of the sort is $O(1)$. Hence, the space complexity of Algorithm (\ref{alg::LB-OPAR}) is $O(|V|+|E|)$. 
\end{proof}

It is worthy to mention that we have two reasons in choosing our LP-based methodology instead of workflow scheduling techniques such as the reference one multi-objective heterogeneous earliest-finish-time  (MOHEFT) \cite{MOHEFT}, and the much improved one FAst  workflow  scheduling  approach based  MOHEFT  using BAcktraking and CHeckpointing  (FAMOBACH) \cite{FAMOBACH}. Workflow scheduling techniques, in general, work with much lower computational complexity than LP algorithms and lead to near-optimal solutions. However, our proposed solution finds the exact optimal answer to the problem as we proved it in Lemma (\ref{lem::opt}),  and the computational complexity of our solution is $O(|E|^2)$ at its worst case, as we shown in Lemma (\ref{lem::complexity}), which is fairly low in comparison with the mentioned alternatives.

\subsection{Link Lifetime Prediction and Node Load Calculation}
\label{sec::LLTprediction}

The optimization problem of Section (\ref{sec::optimization}) considers a $n\times n$ matrix $\mathcal{T}$ which contains link lifetimes as well as vector $\mathcal{L}$ with $n$ elements which contain nodes' load. In this section, we describe how we predict the link lifetimes and calculate node load to form $\mathcal{T}$ and $\mathcal{L}$. We start with $\mathcal{T}$ where we need only three consecutive node locations to fill it. Each node regularly sends a packet to the controller, including its last three consecutive locations in 3D space. The controller, in its turn, calculates the azimuthal and polar directions of each node as well as their speed and acceleration. Using this information, the controller calculates the time for each pair of neighboring nodes that they go out of the communication range of each other. To use simpler notations and without the loss of generality, we consider the communication range of all nodes equal and represent it with $\Re$. It is worthy to mention that our predicted lifetimes stay precise if the corresponding nodes do not change their direction or acceleration. 

We represent the location of node $u_i$ at time $t$ by the tuple $(x_i(t),y_i(t),z_i(t))$. We show the distance traversed by the node $u_i$ in the time interval $[t_1 \quad t_2]$ by $d_{(t_1,t_2)}^{(i)}$ and it could be easily calculated by Equation (\ref{eq::distance}). 

\begin{equation}
\label{eq::distance}
\scriptstyle 
    d_{(t_1,t_2)}^{(i)}=\sqrt{[x_i(t_2)-x_i(t_1)]^2 + [y_i(t_2)-y_i(t_1)]^2 + [z_i(t_2)-z_i(t_1)]^2}
\end{equation}

Considering the three consecutive node locations at times $t_0$, $t_1$, and $t_2$, we calculate the azimuthal angle $\alpha_i$ and polar angle $\theta_i$ of the node's movement direction using Equations (\ref{eq::azimuthal}) and (\ref{eq::polar}), respectively.

\begin{equation}
\label{eq::azimuthal}
\alpha_i=tan^{-1}({\frac{y_i(t_2)-y_i(t_1)}{x_i(t_2)-x_i(t_1)}})\\
\end{equation}

\begin{align}
\label{eq::polar}
\theta_i&=cos^{-1}({\frac{z_i(t_2)-z_i(t_1)}{d_{(t_1,t_2)}^{(i)}}})\\\nonumber
&=tan^{-1}({\frac{\sqrt{[x_i(t_2)-x_i(t_1)]^2 + [y_i(t_2)-y_i(t_1)]^2 }}{z_i(t_2)-z_i(t_1)}})
\end{align}

We can calculate speed $v_i$ and acceleration $a_i$ of node $u_i$ respectively using Equations (\ref{eq::3dSpeed}) and (\ref{eq::acc}).

\begin{equation}
\label{eq::3dSpeed}
v_i(t_1)= \frac{\sqrt{d_{(t_0,t_1)}^{(i)}}}{t_1-t_0}
\end{equation}

\begin{equation}
\label{eq::acc}
a_i(t_2)=\frac{v_i(t_2)-v_i(t_1)}{t_2-t_0}
\end{equation}

Next, we can calculate the predicted location of the node $u_i$ at time $(t_2+\Delta t)$ using Equation (\ref{eq::3dNextPosAcc}). Then, we can calculate the distance between nodes $u_i$ and $u_j$ at time $(t_2+\Delta t)$ using Equation (\ref{eq::nodeDistance}). We represent this distance by $d_{(i,j)}(t_2+\Delta t)$ where $\Delta t$ is unknown. Knowing the communication range of nodes, i.e. $\Re$, we can calculate the time $\Delta t$ using Equation (\ref{eq::root}). This variable is the time required for the nodes to go out of the communication range of one another and represents the lifetime of the link $e_{(i,j)}$. Equation (\ref{eq::root}) is a simple nonlinear equation with only one unknown variable, i.e. $\Delta t$. This equation can be easily solved with any lightweight numerical method such as the Bisection method. It is worth mentioning that the LB-OPAR is a modular routing algorithm. Hence, the prediction algorithm proposed here, while its accuracy is validated in Section (\ref{sec::predEval}), could be replaced with any other prediction algorithm with no need for any modification in the overall routing
solution. 

{\small
\begin{align}
\label{eq::3dNextPosAcc}
&\begin{cases}
x_i(t_2+\Delta t)=x_i(t_2)+( v_i(t_2) \Delta t+\frac{1}{2}a_i \Delta t^2 ) sin(\theta_i) cos(\alpha_i)\\
y_i(t_2+\Delta t)=y_i(t_2)+(  v_i(t_2)\Delta t+\frac{1}{2}a_i \Delta t^2 ) sin(\theta_i) sin(\alpha_i)\\
z_i(t_2+\Delta t)=z_i(t_2)+( v_i(t_2)\Delta t+\frac{1}{2}a_i \Delta t^2 ) cos(\theta_i) 
\end{cases}
\end{align} }

\begin{equation}
\label{eq::nodeDistance}
d_{(t_2+\Delta t)}^{(i,j)}=\\
\sqrt{ \begin{aligned} 
&(x_i(t_2+\Delta t)-x_j(t_2+\Delta t))^2+\\
&(y_i(t_2+\Delta t)-y_j(t_2+\Delta t))^2+\\
&(z_i(t_2+\Delta t)-z_j(t_2+\Delta t))^2
\end{aligned}
}
\end{equation}
\begin{equation}
\label{eq::root}
d_{(t_2+\Delta t)}^{(i,j)}-\Re=0
\end{equation}

In a wireless network, a node can start transmitting only when the channel is empty, which happens when no other node in the communication range of the sending node transmits a packet. If, for any reason, more than one node attempt to send a packet simultaneously, a collision will happen. Generally, wireless networks, via their medium access control (MAC) protocols, try to avoid the collision by preventing one of the nodes from sending its own packets at the same time. This fact leads to a delay in flow completion time. In our SDN-based system, the controller is aware of all transmitting flows and their paths. Hence, it can calculate the number of nodes that may attempt to send a packet in the communication range of a sending node. We refer to this number as the node's load.

To calculate each node's load, the controller starts the network operation with all zeros for the nodes' load. Any time the controller receives a route request, it calculates the optimal path and determines the nodes affected by the new route. Any node included in the new route and all its one hop neighbors are considered as the nodes affected by the new route. Since the controller potentially knows all nodes' locations, it can also calculate the distance between each pair of nodes, and it further knows the communication range of each node. Hence the controller can determine the affected nodes by the new route. The controller then increases the load of the determined nodes by one. Every time that a route becomes unavailable, the controller decreases the affected nodes' load by one. In two cases, a route is considered unavailable. First, the source node does not receive the TCP acknowledgment packets, and consequently, the source node orders a reroute. Second, the source node receives the acknowledgment of the last packet of the file and sends a flow termination message to the controller.    

\section{Performance Evaluation}
\label{sec:NumericalResults}
To evaluate the performance of LB-OPAR, we use ns-3 network simulator. Ns-3 is a well-known discrete-event simulator widely used for research and development \cite{ns3}. We simulate the network for different scenarios with a different number of nodes and a different number of concurrent file-transfer flows to exhaustively evaluate the network performance in different situations. Network throughput, flow completion time, and flow success rate are measured as the performance evaluation metrics. To show the significant improvement of LB-OPAR in the network performance, we compare its results with the baseline benchmarks AODV, OLSR, DSDV as well as the OPAR algorithm. We select the baseline algorithms that cover the proactive and reactive, link-state and distance-vector, as well as SDN-based and fully distributed routing algorithms. In this section, we first describe the simulation setting. Next, we evaluate the prediction algorithm proposed in Section (\ref{sec::LLTprediction}). Since LB-OPAR is sensitive to its weights $w_1$, $w_2$, and $w_3$, we then perform an analysis on these weights to find the proper value for them in different network settings. Finally, we represent the results of the comparison among LB-OPAR, AODV, DSDV, OLSR, and OPAR, along with the analysis.      

\subsection{Simulation Setting}

We simulate a UAV network with $n$ nodes in a $2000\times 300\times 50$ 3D space. We perform two sets of simulation scenarios. In the first set of the simulation scenarios, we vary the number of nodes from 50 to 100 in an increment step of 5. In this set of scenarios, five source nodes start sending their data to five corresponding destination UAVs. The generated traffic is FTP over TCP NewReno. A 5 MB data file is sent in each flow. In the second set of scenarios, we set the number of nodes equal to fifty, i.e. $n=50$, and vary the number of concurrent flows from one to ten. Each simulation instance is run for 500 seconds. UAV nodes move based on two different random movement patterns, 3D random waypoint, and 3D Gauss-Markov model. While random waypoint chooses its waypoints in a completely random manner, the Gauss-Markov model is closer to pre-planned scenarios. It prevents the sudden significant direction change by applying a limitation on the variability of the movement angle. We choose these two mobility models to show that LB-OPAR shows excellent performance in both random and pre-planned scenarios. The speed range is set to $[0 \quad 50]$ $m/s$. We run each simulation instance 20 times with different random seeds, measure the parameters of interest, and report the average results.

We set the transmission power to 7.5 dBm for each UAV. We use the free-space propagation loss model, also known as Friis. The propagation loss of UAV transmission is very well-known to behave close to the Friis model. The constant speed propagation model is also selected as the propagation delay model. UAV nodes use IEEE 802.11b as their wireless communication standard. Table (\ref{tbl::simSetting}) represents the simulation parameters.

\begin{table}[t]
\caption[]{Simulation Setting }
\resizebox{1\textwidth}{!}{
\begin{minipage}{\textwidth}
\begin{tabular}{ l | l  }
 Simulator version & ns-3 3.30\\
  Number of UAVs & $[50\quad 100]$ \\
  Area size & $300\times 2000\times 50 m$\\
  Transmission power & 7.5 dBm\\
  Number of concurrent flows & $[1 \quad 10]$\\
  File size & 5 MB\\
  Simulation time & 500 sec\\
  Speed range & $[0 \quad 50] m/s$\\
  Mobility models & 3D RWP\\
  Traffic type & TCP NewReno\\
  Wireless communication standard & IEEE 802.11b\\
  Propagation loss model & Free-space propagation loss\\ 
  Propagation delay model & Constant speed propagation delay
\end{tabular}
\label{tbl::simSetting}
\end{minipage}}
\end{table}

We measure the flow success rate, network throughput, and flow completion time to make the comparison. We define flow success rate as the rate of the flows that successfully deliver their entire 5 MB file to the corresponding destination node in the simulation time. It is worth mentioning that if the simulation time is infinite, all flows with any of the routing algorithms can deliver their files successfully. However, practically it is not feasible to keep the simulation without stop time. Hence, if a flow destination does not receive all of the 5 MB file's packets in the simulation time, we consider that flow as unsuccessful. In the real world, each communication flow may have a deadline. The simulation time is the mapping of the communication deadline concept in the simulation world.

We use the standard definition for throughput, which is the rate of successful packet delivery over the communication bandwidth, and it is measured in Mbps. While we report the average throughput for all successful flows, we find that it might not be fair to report the average of raw throughput values. In some cases, a specific routing algorithm works worse than the others and fails in finding the proper routes. As an instance, it may successfully deliver just one file out of ten. Assume that the source node and the destination be physical neighbors for the successful flow. In this case, it achieves a close to $1$ throughput, where the other routing algorithms, which may have a perfect flow success rate, will show much lower throughput. Hence, to report a fair comparison, we weight the throughput with its corresponding flow success rate. 

Finally, we define FCT as the time between the first packet of the 5 MB file sent by the source node and the last packet of the same file received by the destination node. Since the simulation time is 500 seconds, some flows may fail in delivering all of their packets to the corresponding destination in the simulation time. Hence, we consider the FCT equal to 500 for the unsuccessful flows as the lower bound for the FCT of those flows. 

\subsection{Prediction Evaluation}
\label{sec::predEval}

To evaluate the prediction algorithm presented in Section (\ref{sec::LLTprediction}), we design a simulation scenario and compare its results with those of the extrapolation algorithm as a standard baseline. We use Newton divided differences extrapolation \cite{newton} to predict node's position in time $t_2+\Delta t$, using three consequent node positions in times $t_0$, $t_1$, and $t_2$. It is worthy to note that all other standard extrapolation algorithms, such as Lagrange extrapolation, result in the same value of the Newton divided differences method \cite{numerical}. We choose this algorithm for its better computation complexity in comparison with the others. Next, we put the predicted position in Equation (\ref{eq::nodeDistance}) to calculate the distance between $u_i$ and $u_j$ at time $t_2+\Delta t$, and form Equation (\ref{eq::root}). Then, we use a lightweight numerical method such as Bisection to find the root and equivalently predict the link lifetime between the corresponding nodes. We measure the prediction error of the link lifetime as the absolute difference between the predicted lifetime and the simulated one. 

We consider a network with a different number of nodes where the nodes start moving for a warmup process of 100 seconds. We perform the prediction and then continue the simulation for extra 500 seconds. We predict the link lifetimes using both the algorithms of Section (\ref{sec::LLTprediction}) and the extrapolation. We then watch the links to measure their simulated lifetimes. Next, We calculate the expected value of the error for both methods in each scenario. Obviously, we predict the lifetime for the links that already existed at the time of prediction. We consider the lifetime of the links that last alive until the end of the simulation as the remaining simulation time, i.e. 500 seconds. We further consider the lifetime of the nodes in which the prediction returned a value more than the simulation time, equal to the remaining simulation time. We run each scenario 1000 times and calculate the expected value of the error and its standard deviation. While the expected value of the error shows the accuracy of the prediction algorithm, the standard deviation shows its precision. The precision represents the interval of the error distribution around the expected value. We find that the curves of all network densities behave almost identical. Hence, we report only the results for a network with 50 nodes. 

Fig. (\ref{fig::accuracy}) shows the cumulative distribution function (CDF) for the results of the expected value and the standard deviation of the error for both prediction algorithms presented in Section (\ref{sec::LLTprediction}) and the one use extrapolation. Fig. (\ref{fig::accuracy}a) shows that the prediction algorithm proposed here is on average about 70 seconds more accurate than the one use extrapolation. Fig. (\ref{fig::accuracy}b) shows the CDF of the standard deviation around the expected value, which represents the prediction precision. The lower standard deviation value shows higher precision, and consequently, a more reliable prediction algorithm. This figure shows that the standard deviation for the proposed prediction algorithm is averagely around 20 seconds, where this value for the extrapolation algorithm is about 90 seconds. Regardless of the accuracy of the algorithm, the standard deviation shows that the proposed algorithm is more reliable than that of extrapolation since it has fewer fluctuations.  

\begin{figure}[t!]
	\centering
	\subfloat[Error expected value ]{\includegraphics[width=.5\linewidth]{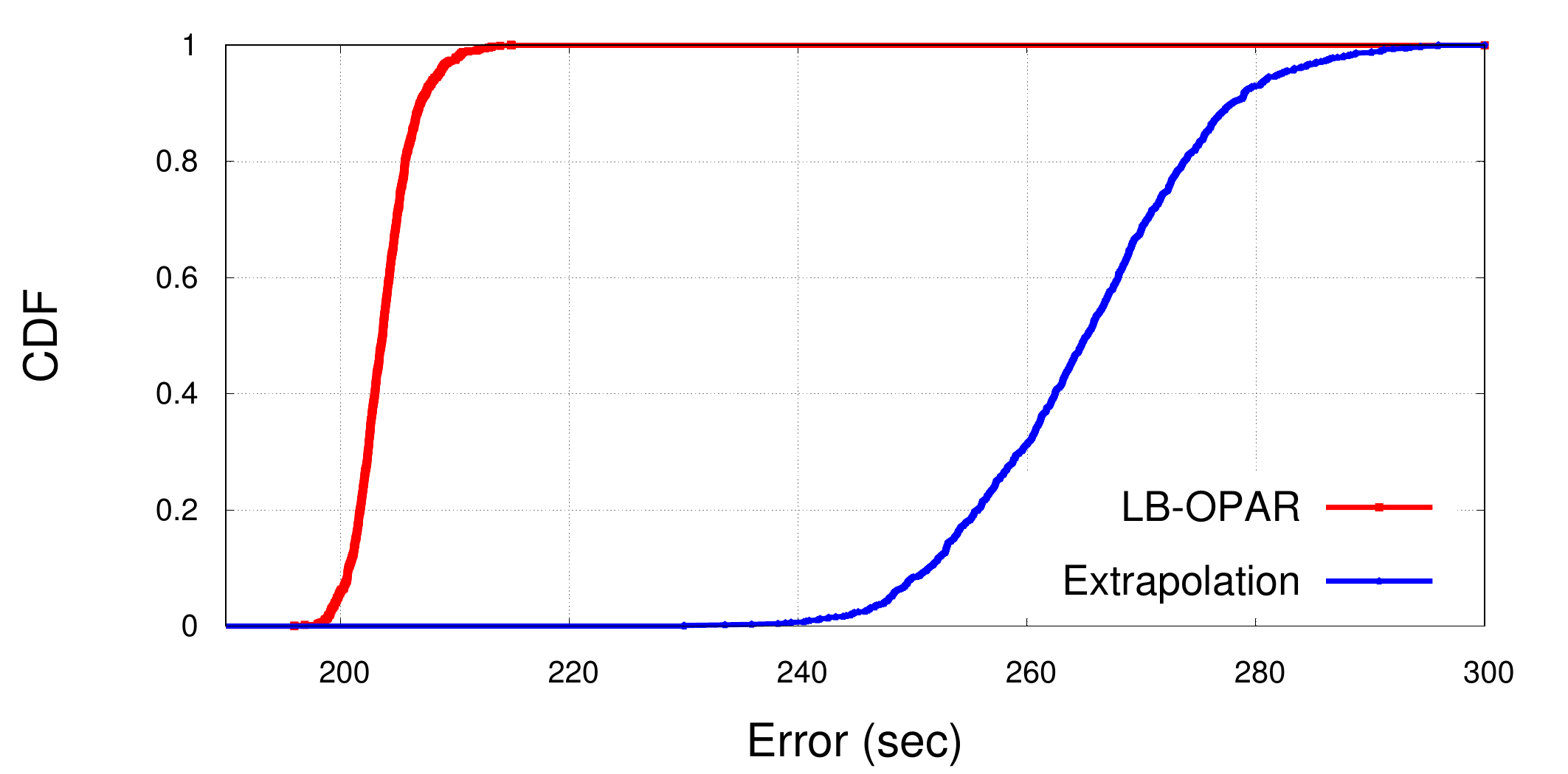}}
	\subfloat[Error standard deviasion]{\includegraphics[width=.5\linewidth]{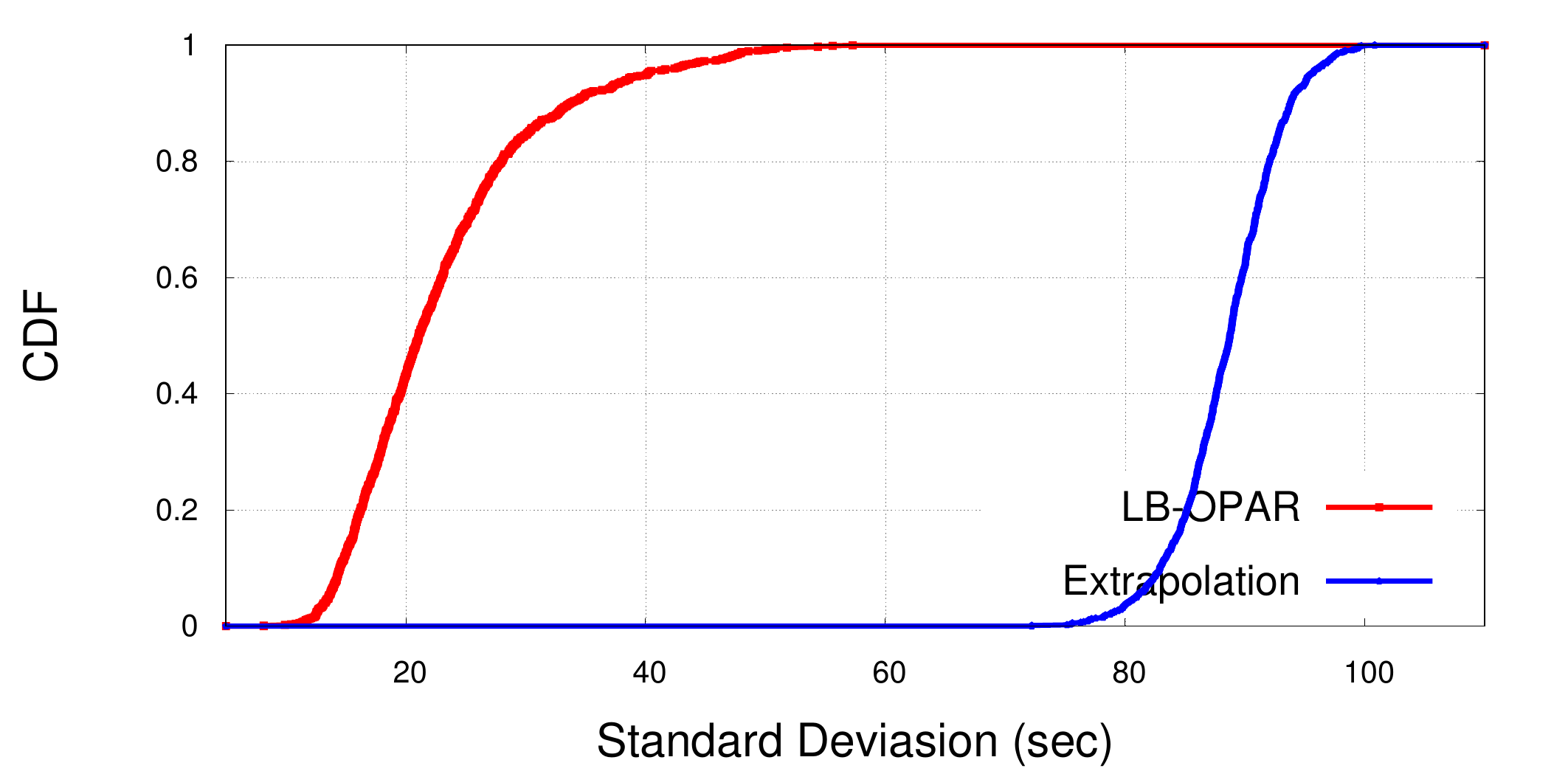}}\\
	\caption{A comparison of CDF of LB-OPAR prediction algorithm versus extrapolation.}
	\label{fig::accuracy}
\end{figure}

\subsection{Weight Analysis}

In this Section, we analyze the path length, path lifetime, and path load weights of the LB-OPAR optimization problem, i.e. $w_i$s. Accordingly, we designed a simulation scenario with the same setting of Table (\ref{tbl::simSetting}) where we vary the weights to measure the sensitivity of the network performance versus this variation. In this simulation, we first analyze the path load weight, i.e. $w_3$. We vary the value of $w_3$ from 0 to 1 in an increment step of $0.1$ and let $w_1=w_2$. Since $\sum_{i=1}^{3}w_i=1$, the value of $w_1$ and $w_2$ could be calculated easily. 

Simulation results show that the highest performance in terms of throughput, success rate, and FCT is achieved for the values of $w_3$ mentioned in Table (\ref{tbl::w3}). The results show that for the network with less load, assigning a value more than zero to $w_3$ might not benefit the network performance. However,  increasing the number of concurrent flows needs a higher weight for $w_3$ to achieve better performance in terms of all mentioned performance metrics. Furthermore, we find that there is no optimal case with $w_3=1$ which means path lifetime and path length always have an obvious effect in choosing the optimal path. We also find that increasing the network density has no significant effect on the value of $w_3$. 

\begin{table}[t]
\caption[]{Path load weight analysis}
\resizebox{\columnwidth}{!}{
\begin{minipage}{\textwidth}
\begin{tabular}{ | l || l | l | l | l | l | l | l | l | l | l | l | }
\hline
 no. of Flows & 1& 2 &3 &4 &5 &6&7&8&9&10\\\hline
 $w_3$& 0& 0 &0 &0.1 &0 &0.3&0.4&0.7&0.6&0.7\\\hline\hline
  no. of UAVs & 55& 60 &65 &70 &75 &80&85&90&95&100\\\hline
 $w_3$        & 0.1& 0 &0.2 &0.1 &0.2 &0&0&0&0&0.1\\\hline
\end{tabular}
\label{tbl::w3}
\end{minipage}}
\end{table}

Next, we analyze the effect of varying path-length weight, i.e. $w_1$, versus path lifetime weight, i.e. $w_2$, for the fixed path load weight represented in Table (\ref{tbl::w3}). By fixing the value of $w_3$ and varying the values of $w_1$ and $w_2$ we find that the best achievable performance is for the values represented in Tables (\ref{tbl::wsCon}) and (\ref{tbl::wsNode}) for different network loads and different network densities, respectively. Table (\ref{tbl::wsCon}) shows that for the network with less number of concurrent communication flows, i.e. less traffic, the path lifetime weight, i.e. $w_2$ is more important in comparison with path length. However, results show that the weights of path lifetime and path length are almost equal for the network with higher traffic. Table (\ref{tbl::wsCon})  also shows that the network density has no significant effect on the values of path load and path lifetime weights. Hence, for the network with five concurrent communication flows and different numbers of UAVs, the best performance is achievable when $w_1=w_2$, in most cases.


\begin{table}[t]
\caption[]{Weight Analysis for Different Network Loads}
\resizebox{\columnwidth}{!}{
\begin{minipage}{\textwidth}
\begin{tabular}{ | l || l | l | l | l | l | l | l | l | l | l | l | }
\hline
 no. of Flows & 1& 2 &3 &4 &5 &6&7&8&9&10\\\hline
 $w_3$& 0& 0&0 &0.1 &0 &0.3&0.4&0.7&0.6&0.7\\\specialrule{.1em}{.05em}{.05em} 
 $w_1$& 0.3 & 0.3    &0.3 &0.2  &0.5   &0.4   &0.4  &0.15   &0.2  &0.15\\
 $w_2$& 0.7 & 0.7 &0.7 &0.7 &0.5 &0.3&0.2&0.15&0.2&0.15\\\hline
  \end{tabular}
\label{tbl::wsCon}
\end{minipage}
}
\end{table}

\begin{table}[t]
\caption[]{Weight Analysis for Different Network Densities}
\resizebox{\columnwidth}{!}{
\begin{minipage}{\textwidth}
\begin{tabular}{ | l || l | l | l | l | l | l | l | l | l | l | l | l | }
\hline
 no. of UAVs & 50& 55 &60 &65 &70 &75 & 80 &85 &90&95&100\\\hline
 $w_3$& 0& 0.1 &0 &0.2 &0.1 &0.2&0&0&0&0 &0.1\\\specialrule{.1em}{.05em}{.05em} 
 $w_1$& 0.5 & 0.45    &0.5 &0.4  &0.6   &0.4   &0.5  &0.4   &0.4  &0.5 &0.45\\
 $w_2$& 0.5 & 0.45 &0.5 &0.4 &0.3 &0.4&0.5&0.6&0.6&0.5 &0.45\\\hline
  \end{tabular}
\label{tbl::wsNode}
\end{minipage}
}
\end{table}

\subsection{Simulation Results}

In this section, we compae the LB-OPAR simulation results with those of OPAR, AODV, DSDV, and OLSR, for both RWP and G-M mobility models. We compare the results for flow success rate, network throughput, and flow completion time for different network densities and different network loads, i.e. concurrent flows.

Fig. (\ref{fig::success}) shows the results for the flow success rate. DSDV and OLSR routing protocols show less than $10\%$ of flow success rate, in most cases, which is mainly because of their inability in catching up with the
high dynamicity of the nodes' movement in 3D space. AODV, however, represents better performance in terms of flow success rate in comparison with DSDV and OLSR. It successfully delivers $30\%$ of the flows, in average. Indeed, both OPAR and LB-OPAR significantly outperform the other algorithms by the average of $60\%$ flow success rate. As it is obvious, increasing the network load leads to more significant performance improvement for LB-OPAR compared to OPAR by the average of $5\%$. For the less number of concurrent flows, since the network load weight is zero or close to zero, no significant difference is shown between LB-OPAR and OPAR, as it is the case for the network with 5 concurrent flows and different densities, i.e. Fig. (\ref{fig::success}c) and (\ref{fig::success}d). Something worth mentioning is that AODV, DSDV, and OLSR show a higher success rate under the RWP mobility model than G-M. It seems that the randomness of the RWP model helps them to increase their performance. A similar fact was previously discovered in different network studies such as \cite{jellyfish}. OPAR and LB-OPAR show no significant difference under RWP and G-M mobility models, representing their stability of performing well in different conditions.        

\begin{figure}[t!]
	\centering
	\subfloat[Different number of  flows (RWP)]{\includegraphics[width=.5\linewidth]{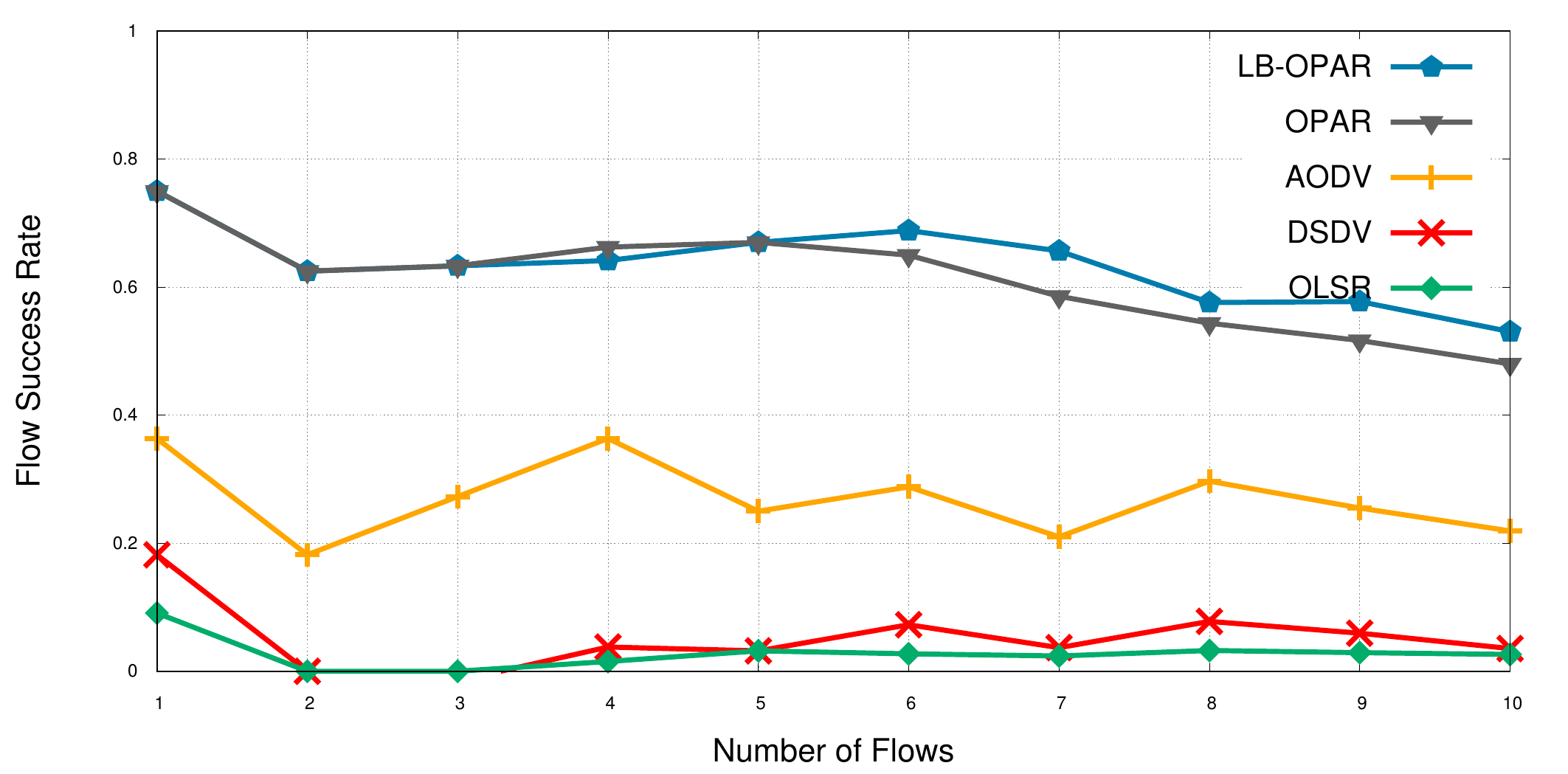}}
	\subfloat[Different number of  flows (G-M)]{\includegraphics[width=.5\linewidth]{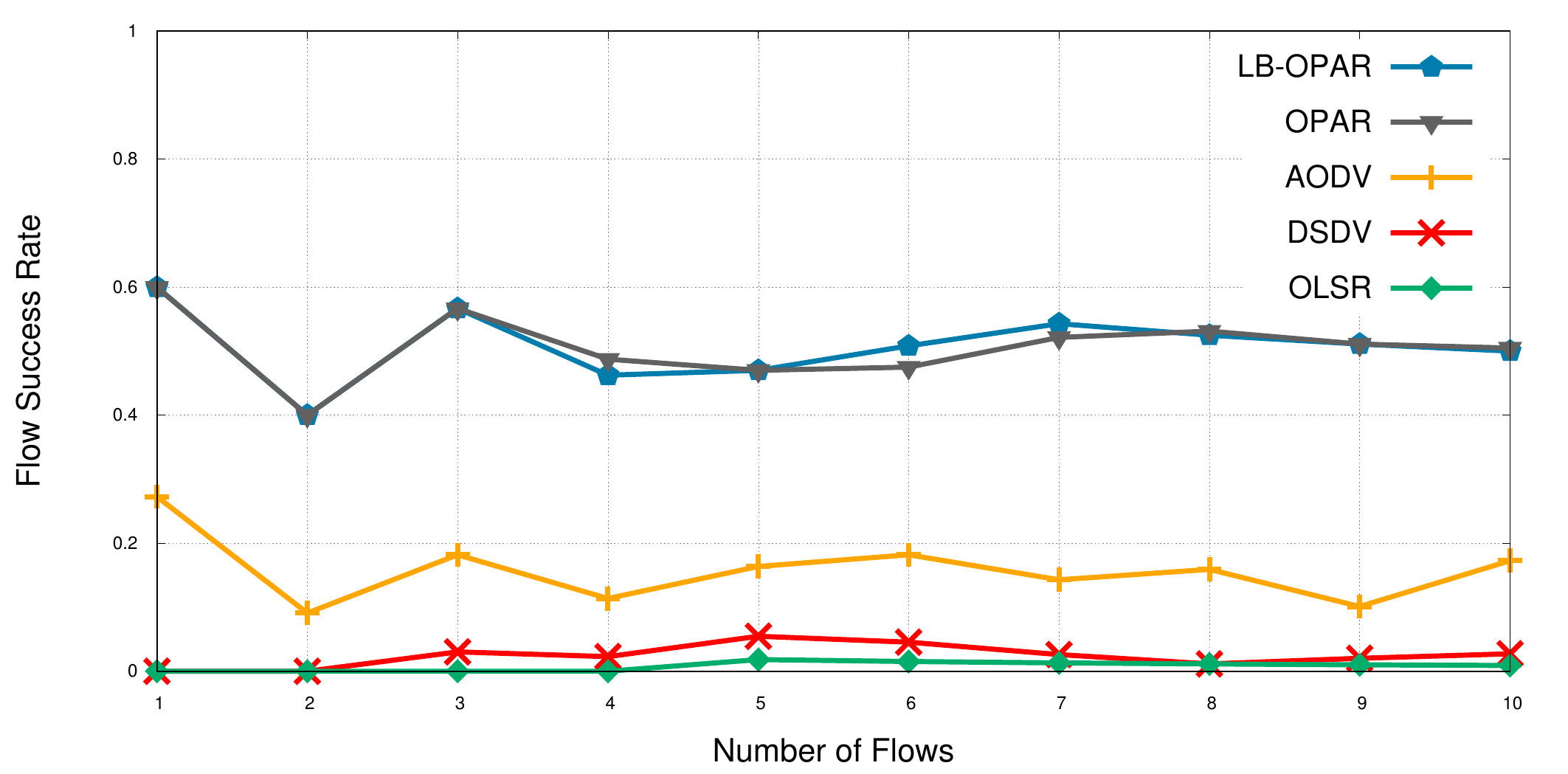}}\\
	\subfloat[Different number of UAVs (RWP)]{  \includegraphics[width=.49\linewidth]{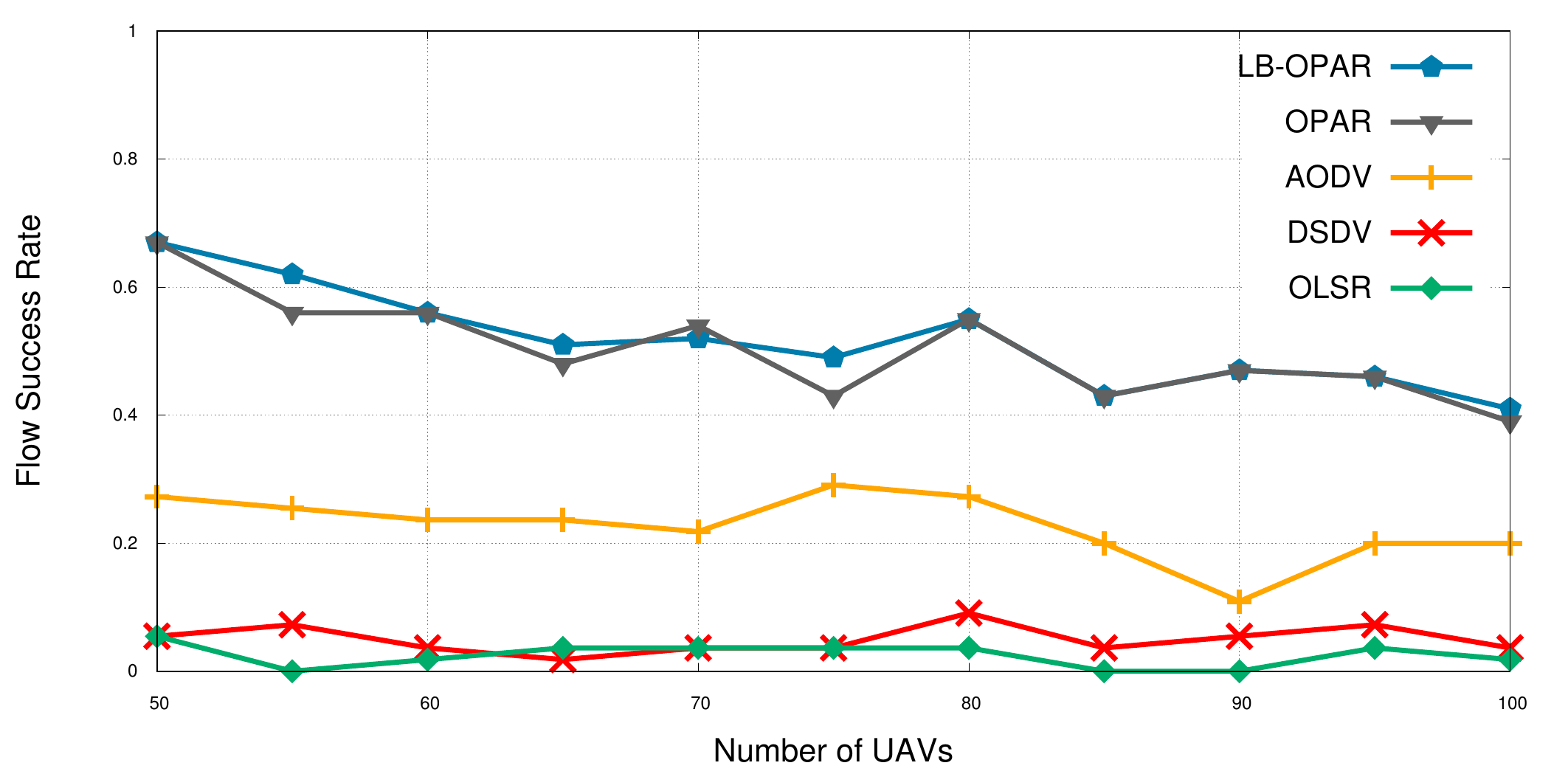}}
	\subfloat[Different number of UAVs (G-M)]{  \includegraphics[width=.49\linewidth]{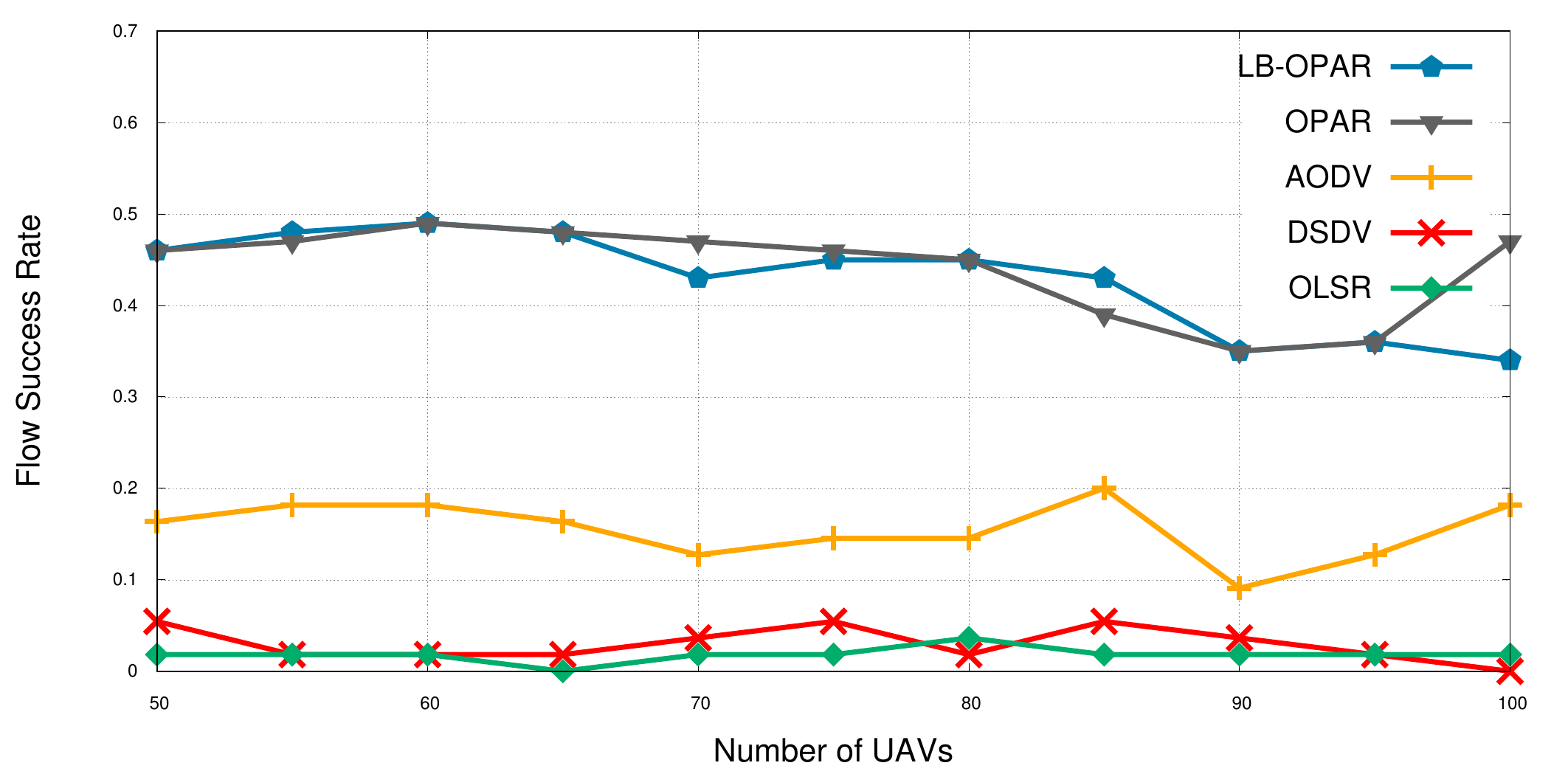}}
	\caption{A comparison of the flow success rate for different loads and densities.}
	\label{fig::success}
\end{figure} 

Next, we investigate the network throughput for all the combinations of different network loads, network densities, and mobility models. Fig. (\ref{fig::throughput}) shows the results where the achieved throughput by DSDV and OLSR is close to zero due to their failure in handling the fast topology changes. Interestingly and in contrast to the conventional mobile networks, increasing the number of concurrent flows shows no significant effect on the throughput, which means in highly dynamic networks, the bottleneck of the throughput is the reroute process due to the topology change not the load of the network. For this parameter, AODV
shows much better performance in comparison with DSDV and OLSR, while OPAR and LB-OPAR outperforms AODV by $100\%$ improvements. OPAR and LB-OPAR show competitive behavior. However, LB-OPAR outperforms OPAR in higher network loads and under the RWP mobility model for about $20\%$. For different network densities, since there are only five concurrent flows and the network load weight is zero or close to zero, we see no significant differences between LB-OPAR and OPAR. However, they outperform AODV by the average of $100\%$ improvements.  

\begin{figure}[t!]
	\centering
	\subfloat[Different number of  flows (RWP)]{\includegraphics[width=.5\linewidth]{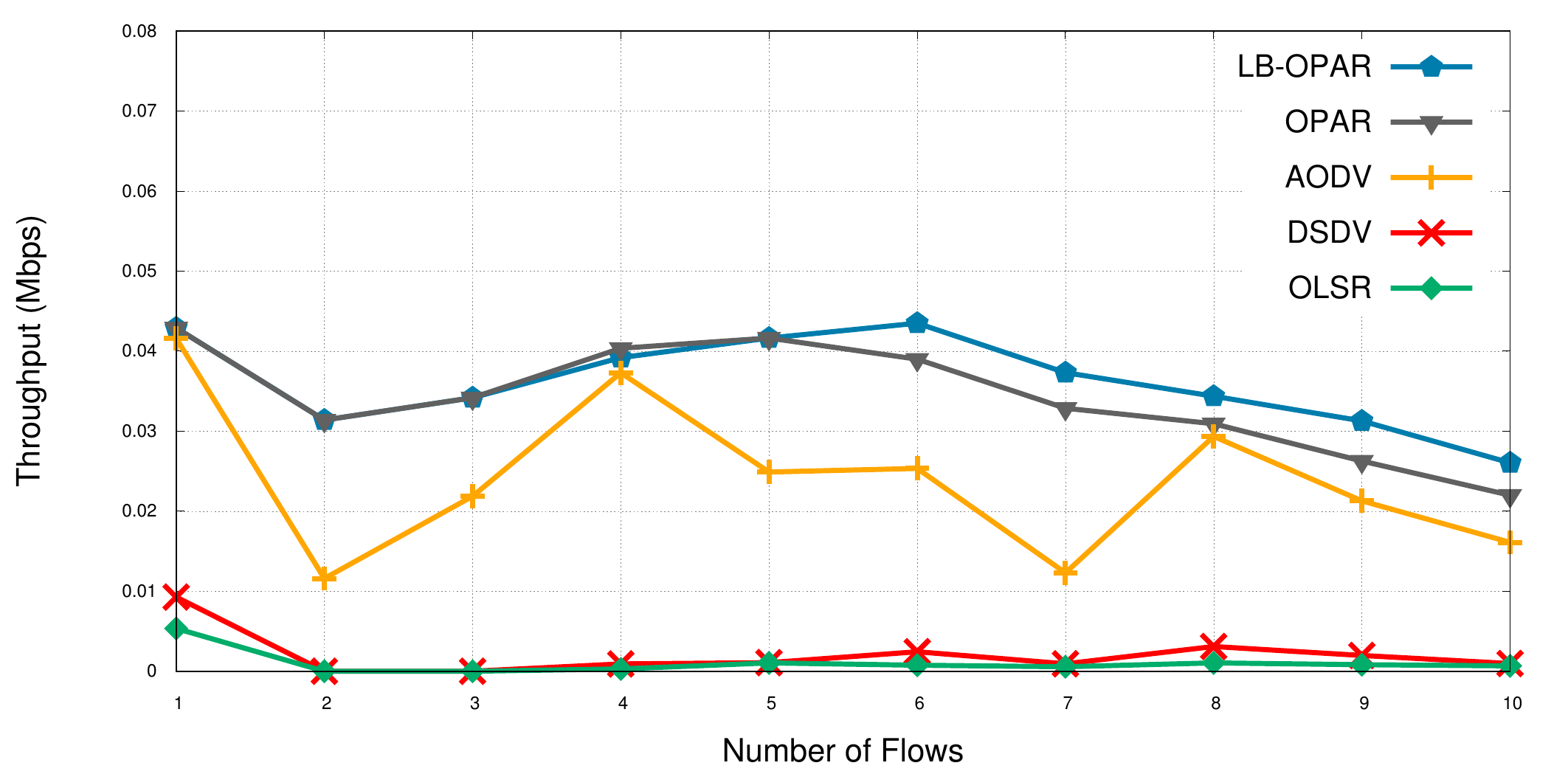}}
	\subfloat[Different number of flows (G-M)]{\includegraphics[width=.5\linewidth]{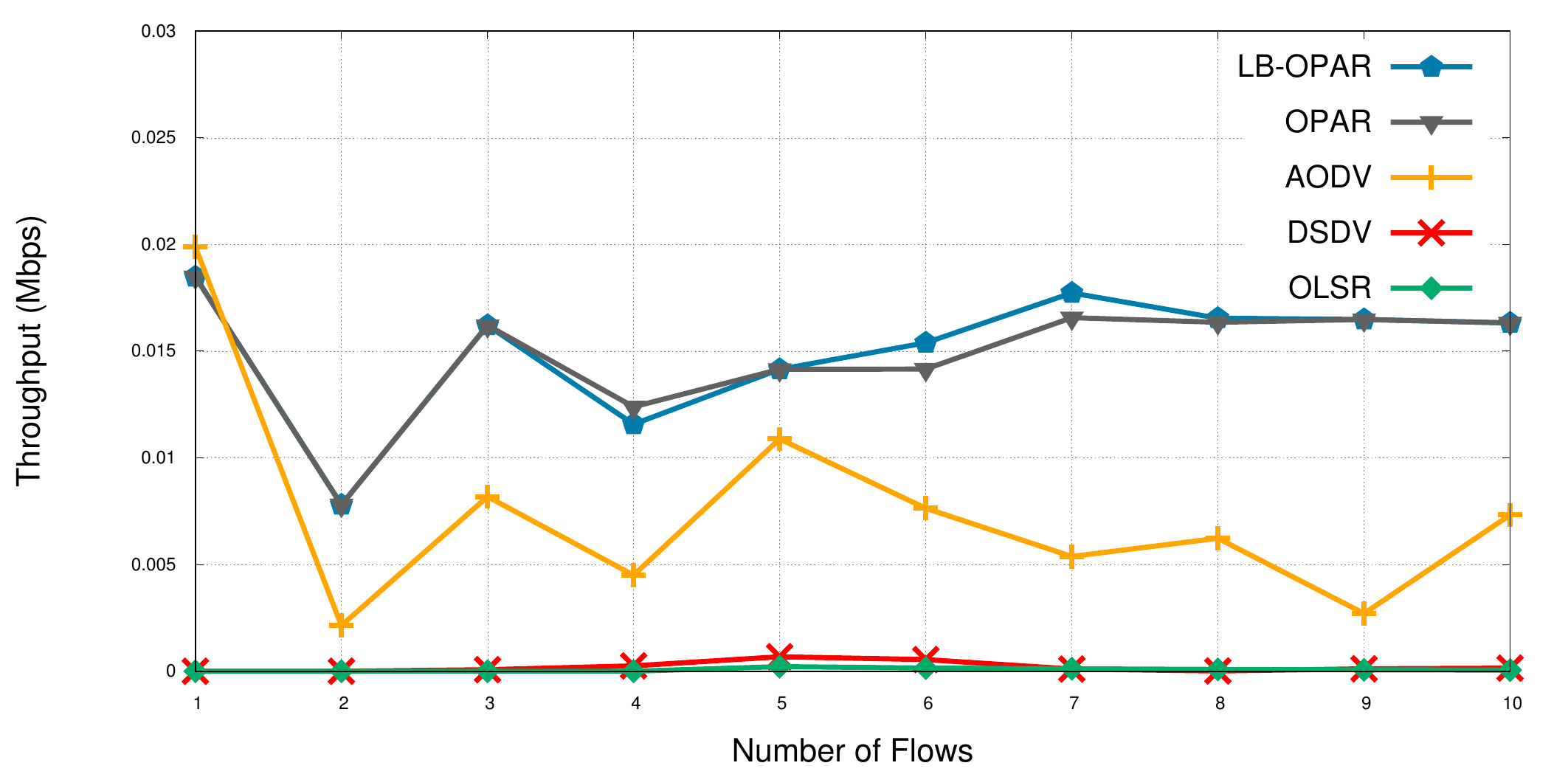}}\\
	\subfloat[Different number of UAVs (RWP)]{  \includegraphics[width=.49\linewidth]{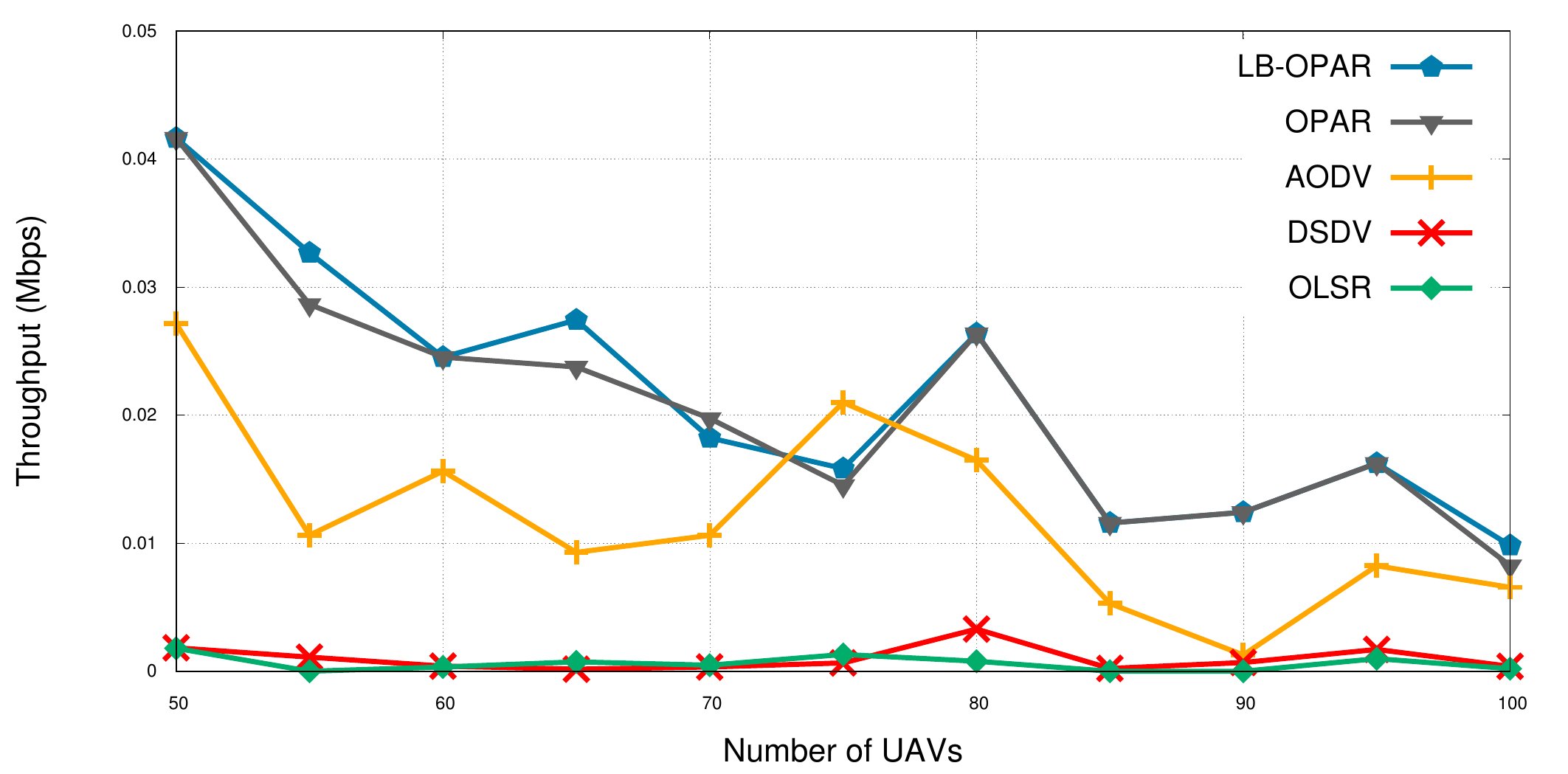}}
	\subfloat[Different number of UAVs (G-M)]{  \includegraphics[width=.49\linewidth]{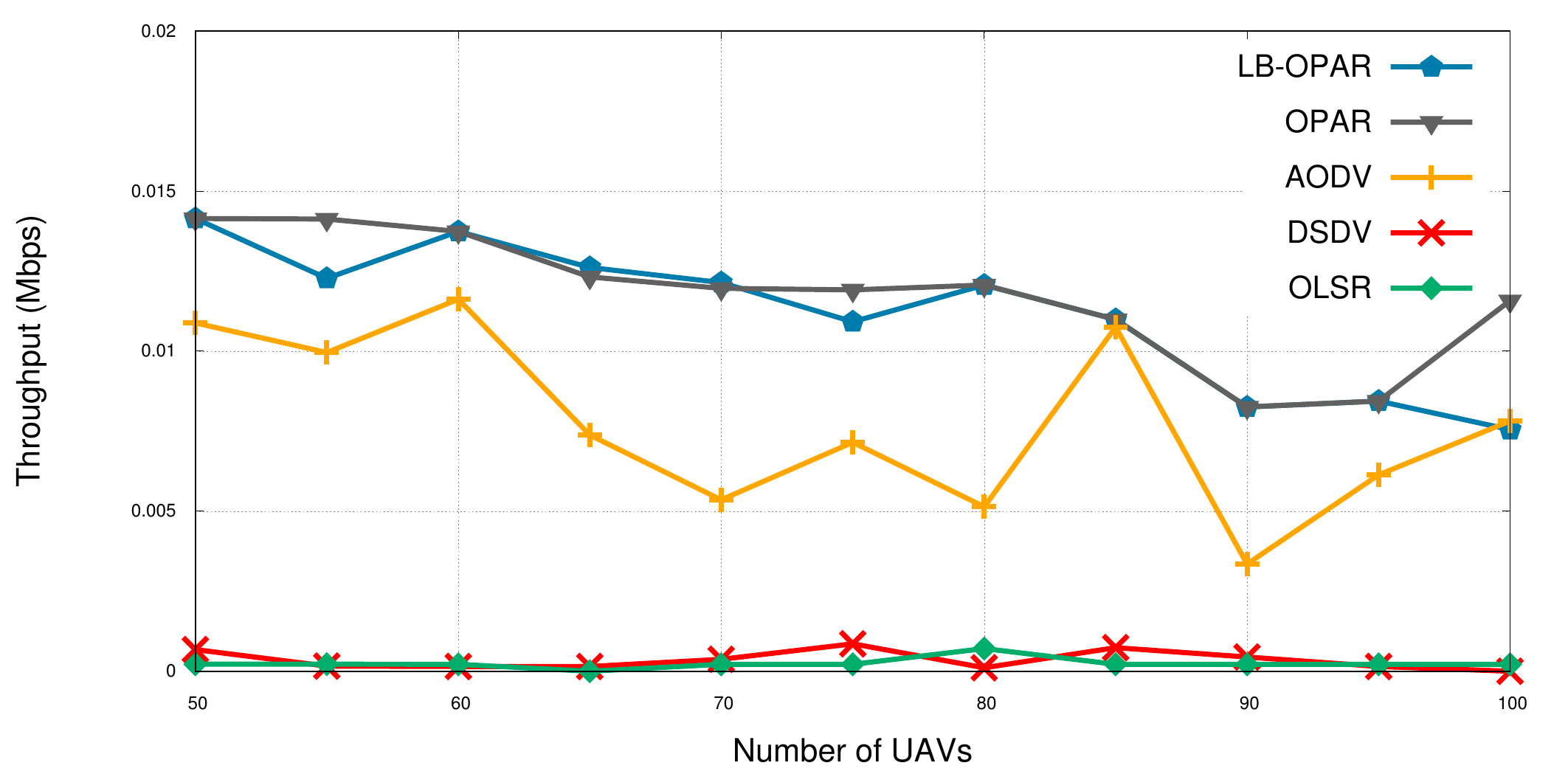}}
	\caption{A comparison of the throughput for different loads and densities.}
	\label{fig::throughput}
\end{figure}

Fig. (\ref{fig::FCT}) shows the results for flow completion time. As we mentioned earlier, we consider the simulation time as the lower bound of the FCT for the failed flows. Hence, the results of Fig. (\ref{fig::FCT}) show the lower bound for the different algorithms' FCT. This figure includes the combination of all network densities, loads, and the simulated mobility models. For all simulation settings,  OPAR and LB-OPAR show the lowest FCT compared to the other routing algorithms. However, LB-OPAR for the network with higher loads shows around $15\%$ less FCT in comparison with OPAR, under RWP mobility model. OPAR and LB-OPAR outperforms AODV by averagely $10\%$.  AODV, in its turn, shows a significantly lower FCT than that of OLSR and DSDV where their average FCT is close to the simulation time, i.e. 500 s, due to theire high failure rate. In many cases, OPAR and LB-OPAR show close FCT results, which is one of two cases. The first case is when there are one or  a few number of concurrent flows which lead to negligible effect of network load in choosing the optimal path. The latter case owing to the fact  that in some cases, there is only one path, or there are several paths, but just one of them has a reasonable objective value, where the others are very long or with a very short lifetime. Hence, both the OPAR and LB-OPAR select the same path, which leads to the close results.      

\begin{figure}[t!]
	\centering
	\subfloat[Different number of  flows (RWP)]{\includegraphics[width=.5\linewidth]{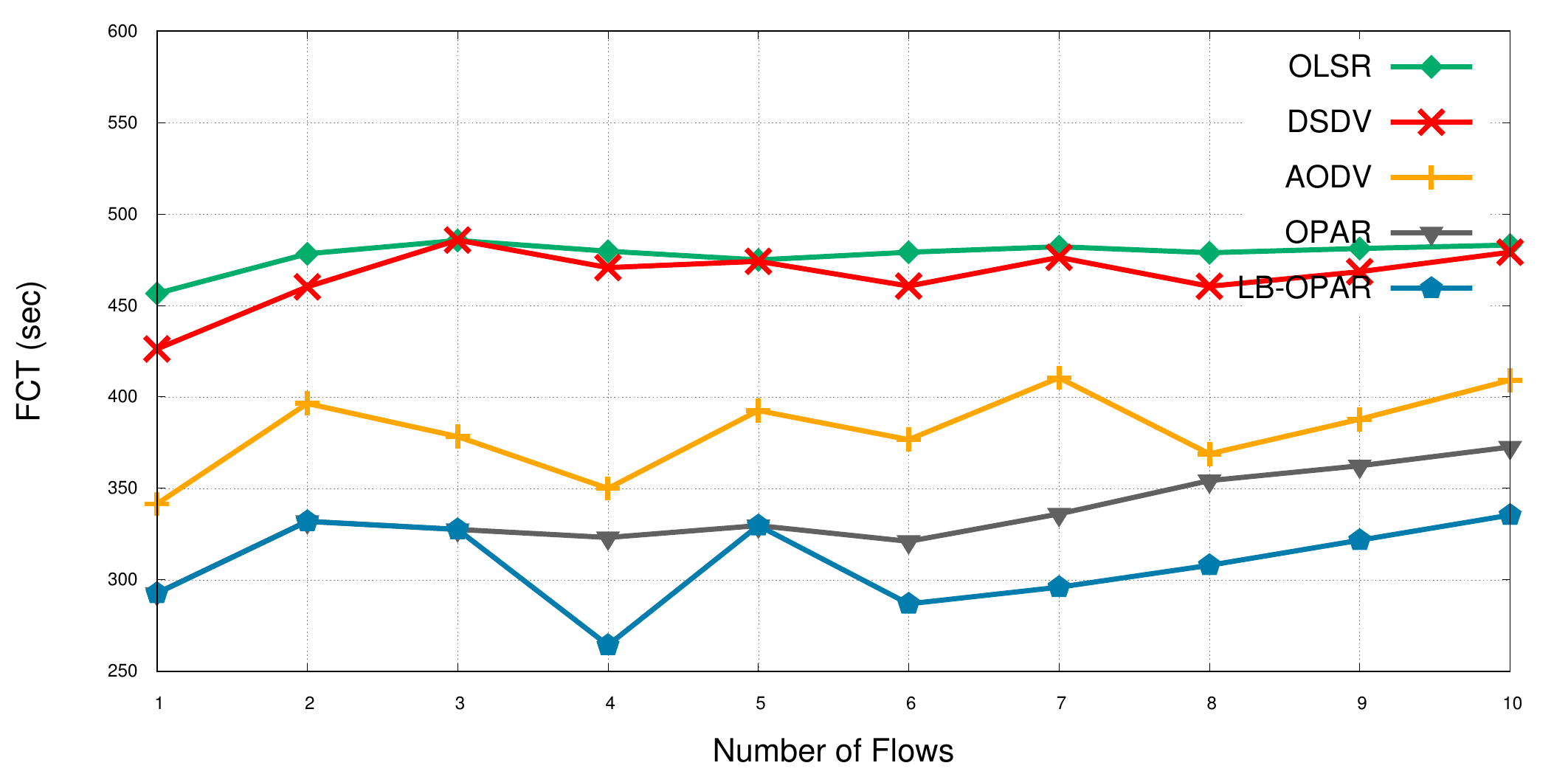}}
	\subfloat[Different number of flows (G-M)]{\includegraphics[width=.5\linewidth]{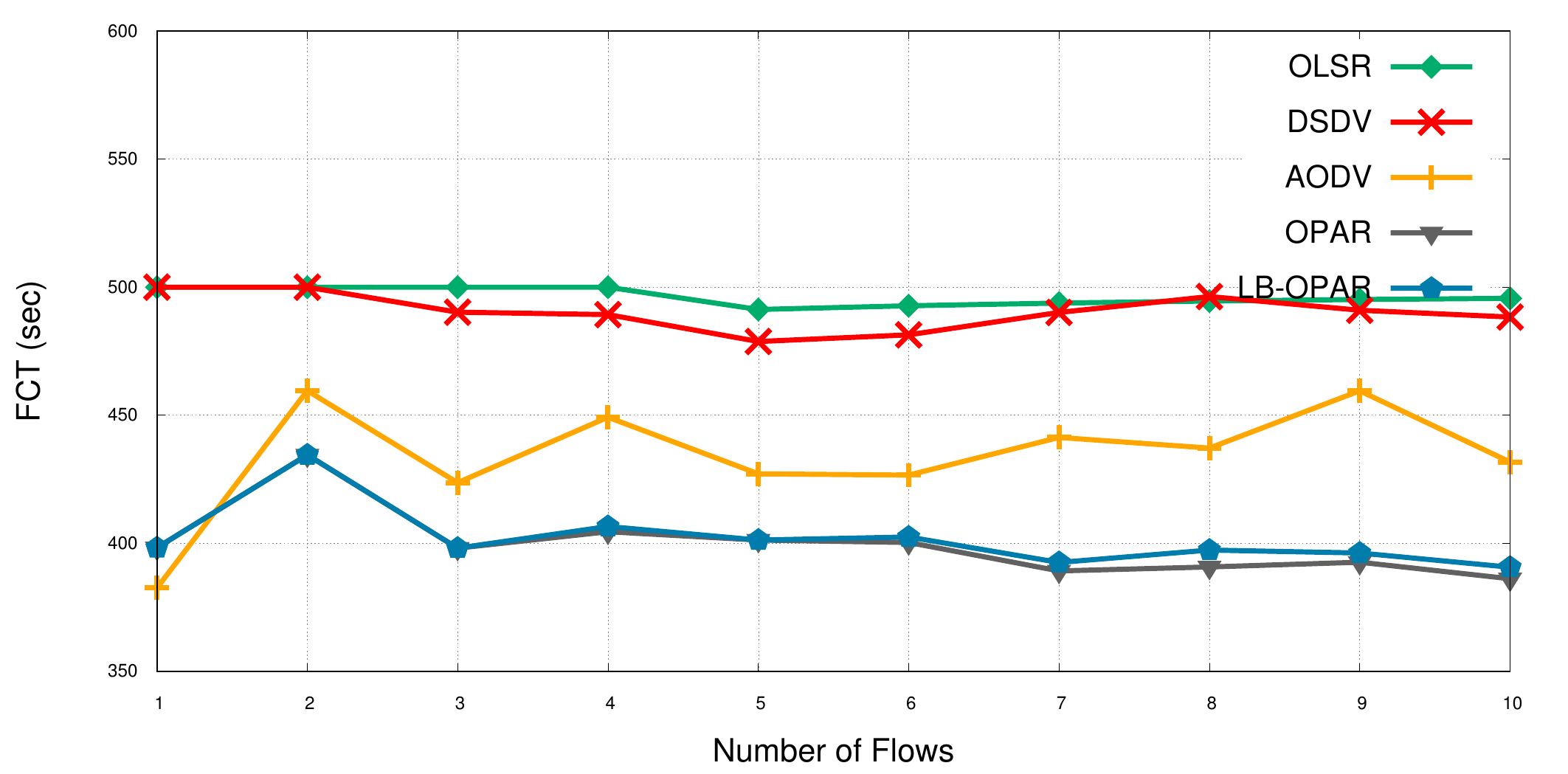}}\\
	\subfloat[Different number of UAVs (RWP)]{  \includegraphics[width=.49\linewidth]{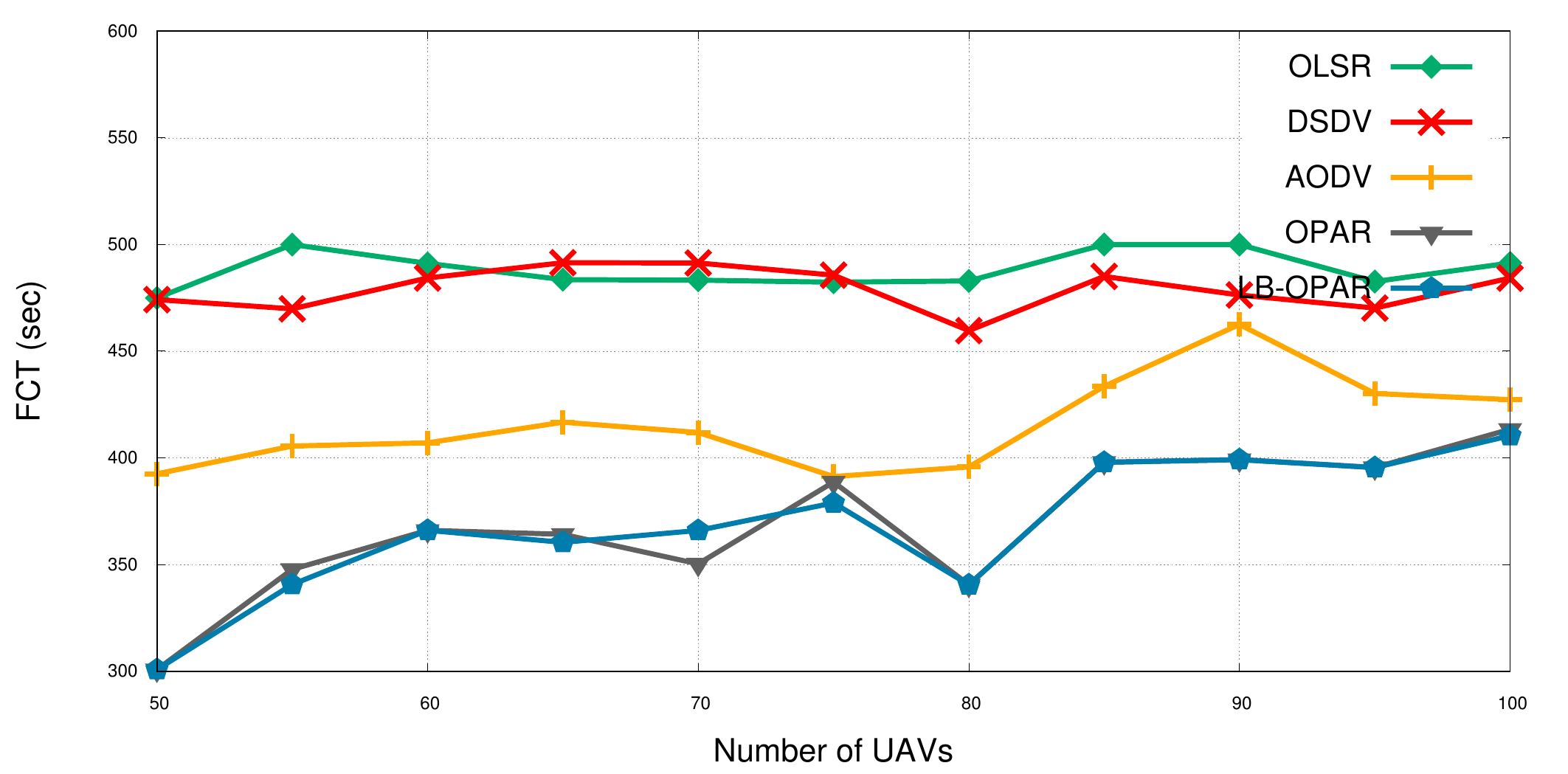}}
	\subfloat[Different number of UAVs (G-M)]{  \includegraphics[width=.49\linewidth]{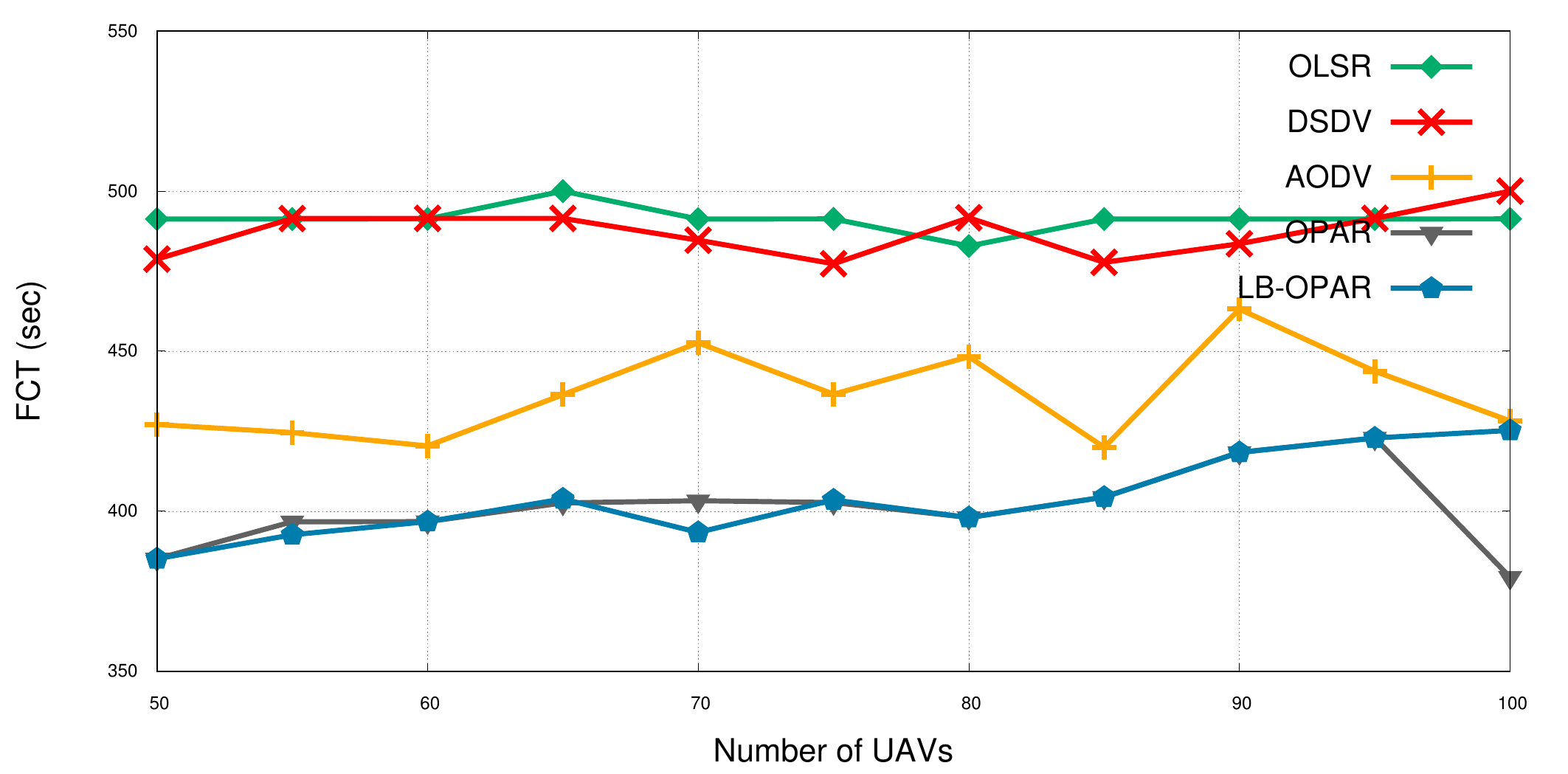}}
	\caption{A comparison of FCT for different loads and densities.}
	\label{fig::FCT}
\end{figure}

\section{Acknowledgement}
This material is based upon work supported by the Air Force Office of Scientific Research under award number FA9550-20-1-0090 and the National Science Foundation under Grant Number CNS-2034218. Any opinions, findings and conclusions or recommendations expressed in this material are those of the author(s) and do not necessarily reflect the views of the US government or AFRL.

\section{Conclusions and Future Work}
\label{sec::conclusion}

The highly dynamic nature of cooperative UAV networks causes the conventional routing algorithms, which aim at finding the shortest path, to fail in reaching the optimal network performance. While we have shown, in our recently published work OPAR, that considering the path lifetime in selecting the routes has a significant impact on the network performance, in this paper, we extend OPAR to consider the network load as well. Accordingly, in this paper, we proposed load-balanced OPAR, an SDN-based routing algorithm for cooperative UAV networks. We modeled the entire problem with an analytical model and proposed a lightweight solution for the proposed model. We implemented LB-OPAR in the ns-3 network simulator and exhaustively evaluated its performance. We show that LB-OPAR outperforms the benchmark routing algorithms AODV, DSDV, and OLSR. It also improves the performance of OPAR when the network works under higher loads. We further found that the performance bottleneck in highly dynamic UAV networks is basically the route breaks down, more than any other reason. Hence, proposing a routing algorithm to cover the multiple reroutes problem has significant importance in improving cooperative UAV networks' performance. Since LB-OPAR has a lightweight computational and space complexity, we aim at distributing the SDN controller tasks among all network nodes for fully distributed networks as future work.   


\nocite{*}
\bibliographystyle{IEEEtran}
\bibliography{References}

\end{document}